%% file: main.tex
\documentclass[journal, twoside]{IEEEtran}
\usepackage[utf8]{inputenc}
\usepackage{tablefootnote}
\usepackage{booktabs}
\usepackage{cite,url}
\usepackage[usenames]{color}
\usepackage{fancyhdr}
\usepackage{amsmath,amssymb,amsthm}
\usepackage{graphicx,psfrag,subcaption,pgfplots}%
\usepackage{listings}%
\usepackage{arydshln}
\usepackage{blkarray,mathtools}
\usepackage{todonotes}
\usepackage{algorithm}
\usepackage{algpseudocode}

\usepackage{subfiles}
\usepackage{blindtext}
\usepackage{psfrag}
\input{math.tex}

\newcommand{\FP}[1]{\mathrm{FP}\left(#1\right)}

\pgfplotsset{compat=1.13}
\DeclareMathAlphabet\mathbfcal{OMS}{cmsy}{b}{n}

\newtheorem{lemma}{Lemma}
\newtheorem{theorem}{Theorem}
\newtheorem{example}{Example}
\newtheorem{corollary}{Corollary}
\newtheorem{property}{Property}
\newtheorem*{remark}{Remark}
\newtheorem{definition}{Definition}

\algtext{EndIf}{\textbf{end}}
\algtext{EndFor}{\textbf{end}}
\algnewcommand\algorithmicto{\textbf{to}}

\begin{document}
\title{Sparse Sampling for Inverse Problems with Tensors}

\author{Guillermo~Ortiz-Jiménez, Mario~Coutino, Sundeep~Prabhakar~Chepuri, and Geert~Leus
\thanks{The authors are with the Faculty of Electrical Engineering, Mathematics and Computer Science, Delft University of Technology, The Netherlands.  Email: g.ortizjimenez@student.tudelft.nl, \{m.coutino;s.p.chepuri;g.j.t.leus\}@tudelft.nl}}

\maketitle

\begin{abstract}
We consider the problem of designing sparse sampling strategies for multidomain signals, which can be represented using tensors that admit a known multilinear decomposition.  We leverage the multidomain structure of tensor signals and propose to acquire samples using a Kronecker-structured sensing function, thereby circumventing the curse of dimensionality.  For designing such sensing functions, we develop low-complexity greedy algorithms based on submodular optimization methods to compute near-optimal sampling sets. We present several numerical examples, ranging from multi-antenna communications to graph signal processing, to validate the developed theory.
\end{abstract}

\begin{IEEEkeywords}
Graph signal processing, multidimensional sampling, sparse sampling, submodular optimization, tensors 
\end{IEEEkeywords}

\IEEEpeerreviewmaketitle

\input{introduction}

\input{problem_modelling}

\input{dense_sampling}
\input{diagonal_sampling}

\input{results}

\input{conclusions}

\input{appendix}

\bibliographystyle{IEEEtran}
\bibliography{IEEEabrv,main}

\vfill

\end{document}

%% file: math.tex
\newcommand\numberthis{\addtocounter{equation}{1}\tag{\theequation}}

\DeclareMathOperator*{\argmax}{arg\,max}

\DeclareMathOperator*{\maximize}{maximize\quad}
\DeclareMathOperator*{\MSE}{MSE}
\DeclareMathOperator*{\optimize}{optimize\quad}

\renewcommand{\vec}[1]{\boldsymbol{\mathrm{#1}}}

\newcommand{\R}{\mathbb{R}}
\newcommand{\C}{\mathbb{C}}

\newcommand{\matx}[1]{\boldsymbol{\mathrm{#1}}}

\renewcommand{\det}[1]{\mathrm{det}\left\{#1\right\}}

\newcommand{\tr}[1]{\mathrm{tr}\left\{#1\right\}}

\newcommand{\norm}[1]{\left\lVert#1\right\rVert}

\newcommand{\braket}[1]{\left\langle#1\right\rangle}

%% file: introduction.tex
\section{Introduction}\label{sec:intro}
\IEEEPARstart{I}{n} many engineering and scientific applications, we frequently encounter large volumes of multisensor data, defined over multiple domains, which are complex in nature. For example, in wireless communications, received data per user may be indexed in space, time, and frequency. Similarly, in hyperspectral imaging, a scene measured in different wavelengths contains information from the three-dimensional spatial domain as well as the spectral domain. And also, when dealing with graph data in a recommender system, information resides on multiple domains (e.g., users, movies, music, and so on). To process such multisensor datasets, \emph{higher-order tensors} or \emph{multiway arrays} have been proven to be extremely useful. 

In practice, however, due to limited access to sensing resources, economic or physical space limitations, it is often not possible to measure such multidomain signals using every combination of sensors related to different domains. To cope with such issues, in this work, we propose \emph{sparse sampling} techniques to acquire multisensor tensor data.

Sparse samplers can be designed to select a subset of measurements (e.g., spatial or temporal samples as illustrated in Fig.~\ref{fig:1d_sparse sensing}) such that the desired inference performance is achieved. This subset selection problem is referred to as sparse sampling~\cite{foundations}.  An example of this is field estimation, in which the measured field is related to the source signal of interest through a linear model. To infer the source signal, a linear inverse problem is solved. In a resource-constrained environment, since many measurements cannot be taken, it is crucial to carefully select the best subset of samples from a large pool of measurements. This problem is combinatorial in nature and extremely hard to solve in general, even for small-sized problems. Thus, most of the research efforts on this topic focus on finding suboptimal sampling strategies that yield good approximations of the optimal solution \cite{foundations,cvx_sampling,krause2008near,logdet,frame_potential,yu1997sampling,masazade2012sparsity,nonlinear,chepuri2016sparse,chepuri2014sparsity}.

\begin{figure}
    \centering
    \begin{subfigure}{\columnwidth}
        \centering
        \includegraphics[width=0.6\columnwidth]{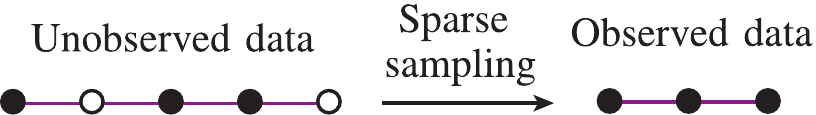}
        \caption{\footnotesize{Single domain sparse sampling}}
    \label{fig:1d_sparse sensing}
    \end{subfigure}
    \vskip 1.5em
    \begin{subfigure}{\columnwidth}
        \centering
        \includegraphics[width=\columnwidth]{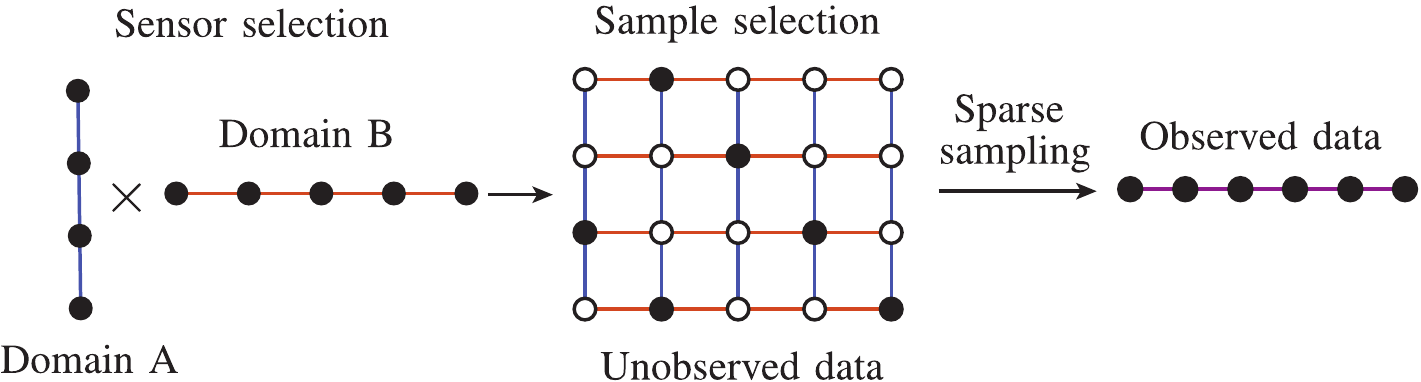}
        \caption{\footnotesize{Unstructured multidomain sparse sampling}}
        \label{fig:unst_sparse_sensing}
    \end{subfigure}
    \vskip 1.5em
     \begin{subfigure}{\columnwidth}
        \includegraphics[width=\columnwidth]{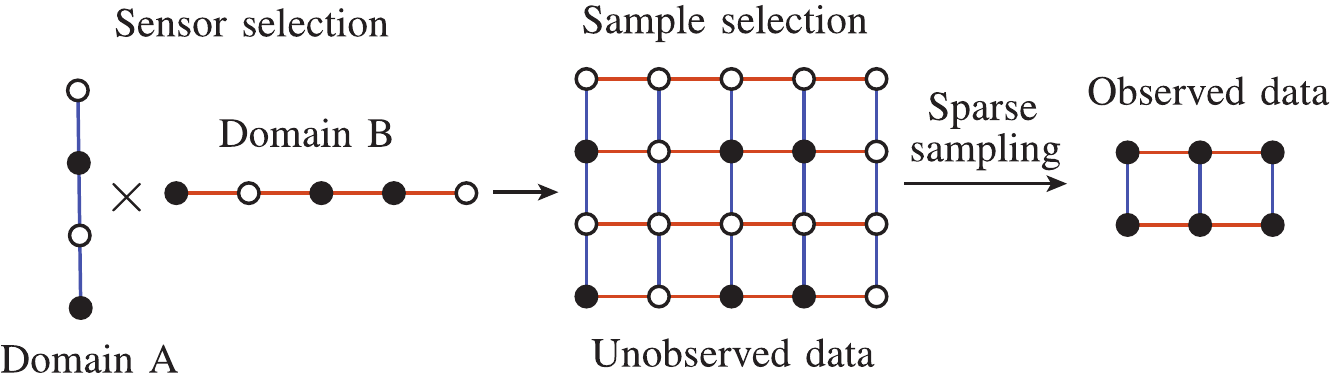}
        \caption{\footnotesize{Kronecker-structured multidomain sparse sampling}}
        \label{fig:str_sparse_sensing}
    \end{subfigure}
     \caption{\small{Different sparse sensing schemes. Black (white) dots represent selected (unselected) measurement locations. Blue and red lines determine different domain directions, and a purple line means that data has a single-domain structure.}}
    \label{fig:sparse_sensing}
    \vspace{-6mm}
\end{figure}
For signals defined over multiple domains, the dimensionality of the measurements grows much faster. An illustration of this ``curse of dimensionality" is provided in Fig.~\ref{fig:unst_sparse_sensing}, wherein the measurements now have to be systematically selected from an even larger pool of measurements. Typically used suboptimal sensor selection strategies are not useful anymore as their complexity is too high; or simply because they need to store very large matrices that do not fit in memory (see Section~\ref{sec:modelling} for a more detailed discussion). Usually, selecting samples arbitrarily from a multidomain signal, requires that sensors are placed densely in every domain, which greatly increases the infrastructure costs. Hence, we propose an efficient \emph{Kronecker-structured} sparse sampling strategy for gathering multidomain signals that overcomes these issues. In Kronecker-structured sparse sampling, instead of choosing a subset of measurements from all possible combined domain locations (as in Fig.~\ref{fig:unst_sparse_sensing}), we propose to choose a subset of sensing locations from each domain and then combine them to obtain multidimensional observations (as illustrated in Fig~\ref{fig:str_sparse_sensing}). We will see later that taking this approach will allow us to define computationally efficient design algorithms that are useful in big data scenarios.  In essence, the main question addressed in this paper is, \emph{how to choose a subset of sampling locations from each domain to sample a multidomain signal so that its reconstruction has the minimum error?}

\section{Preliminaries} \label{sec:preliminaries}

In this section, we introduce the notation that will be used throughout the rest of the paper as well as some preliminary notions of tensor algebra and multilinear systems.

\subsection{Notation}
We use  calligraphic letters such as $\mathcal{L}$ to denote sets, and $|\mathcal{L}|$ to represent its cardinality. Upper (lower) case boldface letters such as $\matx{X}$ ($\matx{x}$) are used to denote matrices (vectors). Bold calligraphic letters such as $\mathbfcal{X}$ denote tensors. $(\cdot)^T$ represents transposition, $(\cdot)^H$ conjugate transposition, and $(\cdot)^\dagger$ the Moore-Penrose pseudoinverse. The trace and determinant operations on matrices are represented by $\tr{\cdot}$ and $\det{\cdot}$, respectively. $\lambda_\text{min}\{\matx{A}\}$ denotes the minimum eigenvalue of matrix $\matx{A}$. We use $\otimes$ to represent the Kronecker product, $\odot$ to represent the Khatri-Rao or column-wise Kronecker product; and $\circ$ to represent the Hadamard or element-wise product between matrices. We write the $\ell_2$-norm of a vector as $\norm{\cdot}_2$ and the Frobenius norm of a matrix or tensor as $\norm{\cdot}_F$. We denote the inner product between two elements of a Euclidean space as $\braket{\cdot,\cdot}$. The expectation operator is denoted by $\mathbb{E}\{\cdot\}$. All logarithms are considered natural. In general, we will denote the product of variables/sets using a tilde, i.e., $\tilde{N}=\prod_{i=1}^R N_i$, or $\tilde{\mathcal{N}_i}=\mathcal{N}_1\times\dots\times\mathcal{N}_R$; and drop the tilde to denote sums (unions), i.e., $N=\sum_{i=1}^R N_i$, or $\mathcal{N}=\bigcup_{i=1}^R\mathcal{N}_i$.

Some important properties of the Kronecker and the Khatri-Rao products that will appear throughout the paper are \cite{kr}: $(\matx{A}\otimes\matx{B})(\matx{C}\otimes\matx{D})=\matx{A}\matx{C}\otimes\matx{B}\matx{D}$; $(\matx{A}\otimes\matx{B})(\matx{C}\odot\matx{D})=\matx{A}\matx{C}\odot\matx{B}\matx{D}$; $(\matx{A}\odot\matx{B})^H(\matx{A}\odot\matx{B})=\matx{A}^H\matx{A}\circ\matx{B}^H\matx{B}$; $(\matx{A}\otimes\matx{B})^\dagger=\matx{A}^\dagger\otimes\matx{B}^\dagger$; and $(\matx{A}\odot\matx{B})^\dagger=(\matx{A}^H\matx{A}\circ\matx{B}^H\matx{B})^\dagger(\matx{A}\odot\matx{B})^H$.
\subsection{Tensors}
A tensor $\mathbfcal{X}\in\C^{N_1\times \dots\times N_R}$ of order $R$ can be viewed as a discretized multidomain signal, with each of its entries indexed over $R$ different domains. 

Using multilinear algebra two tensors $\mathbfcal{X}\in\C^{N_1\times \dots \times N_R}$ and $\mathbfcal{G}\in\C^{K_1\times\dots\times K_R}$ may be related by a multilinear system of equations as 
\begin{equation}
  \mathbfcal{X}=\mathbfcal{G}\bullet_1\matx{U}_1\bullet_2\dots\bullet_R\matx{U}_R, \label{eq:tensor_linear_model}
\end{equation}
where $\{\matx{U}_i\in\C^{N_i\times K_i}\}_{i=1}^R$ represents a set of matrices that relates the $i$th domain of $\mathbfcal{X}$ and the so-called \emph{core tensor} $\mathbfcal{G}$, and $\bullet_i$ represents the $i$th mode product between a tensor and a matrix \cite{tensors}; see Fig.~\ref{fig:dense_core}.  Alternatively, vectorizing \eqref{eq:tensor_linear_model}, we have
\begin{equation}
   \vec{x}=\left(\matx{U}_1\otimes\dots\otimes\matx{U}_R\right)\vec{g},\label{eq:vec_kron}
\end{equation} 
with $\vec{x}=\mathrm{vec}(\mathbfcal{X})\in\C^{\tilde{N}}; \, \tilde{N}= \prod_{i=1}^RN_i$, and $\vec{g}=\mathrm{vec}(\mathbfcal{G}) \in\C^{\tilde{K}}; \, \tilde{K}= \prod_{i=1}^RK_i$.

\begin{figure}[t!]
    \centering
    \begin{subfigure}{\columnwidth}
        \centering
        \includegraphics[height=1in,width=0.8\columnwidth]{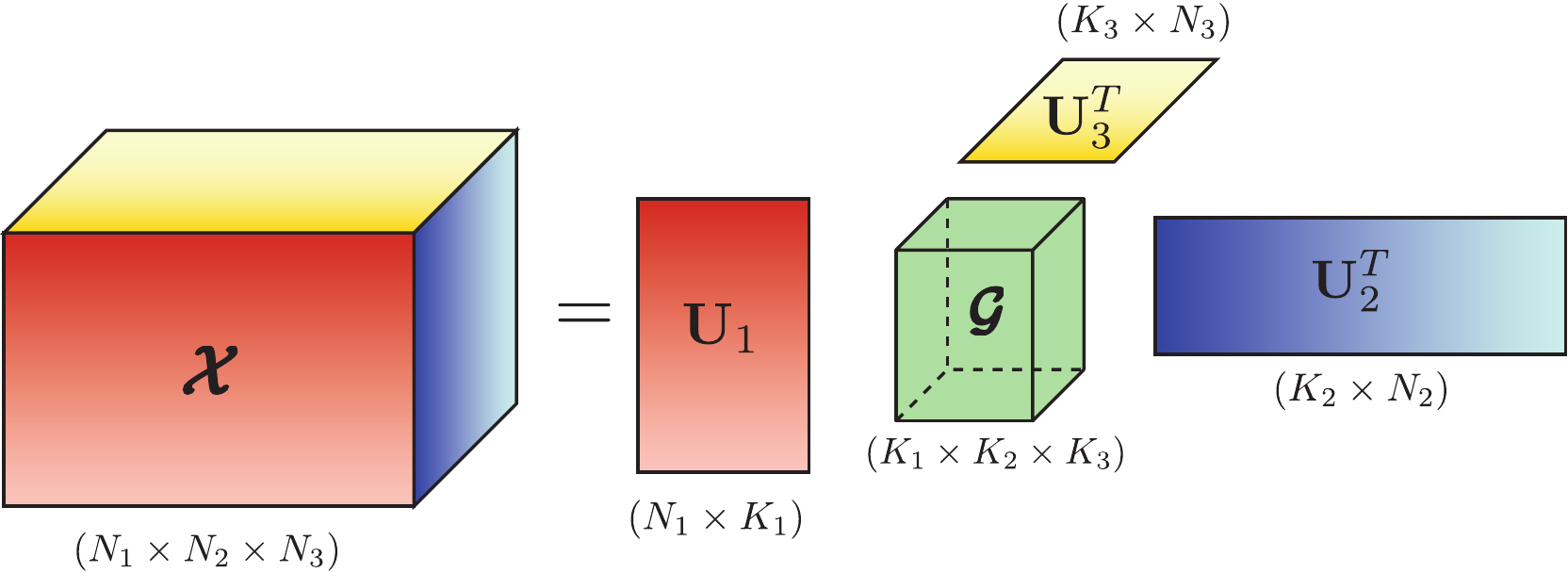}
        \caption{\footnotesize{Dense core}}
    \label{fig:dense_core}
    \end{subfigure}
    \vskip -0.5em
    \begin{subfigure}{\columnwidth}
        \centering
        \includegraphics[height=1in,width=0.8\columnwidth]{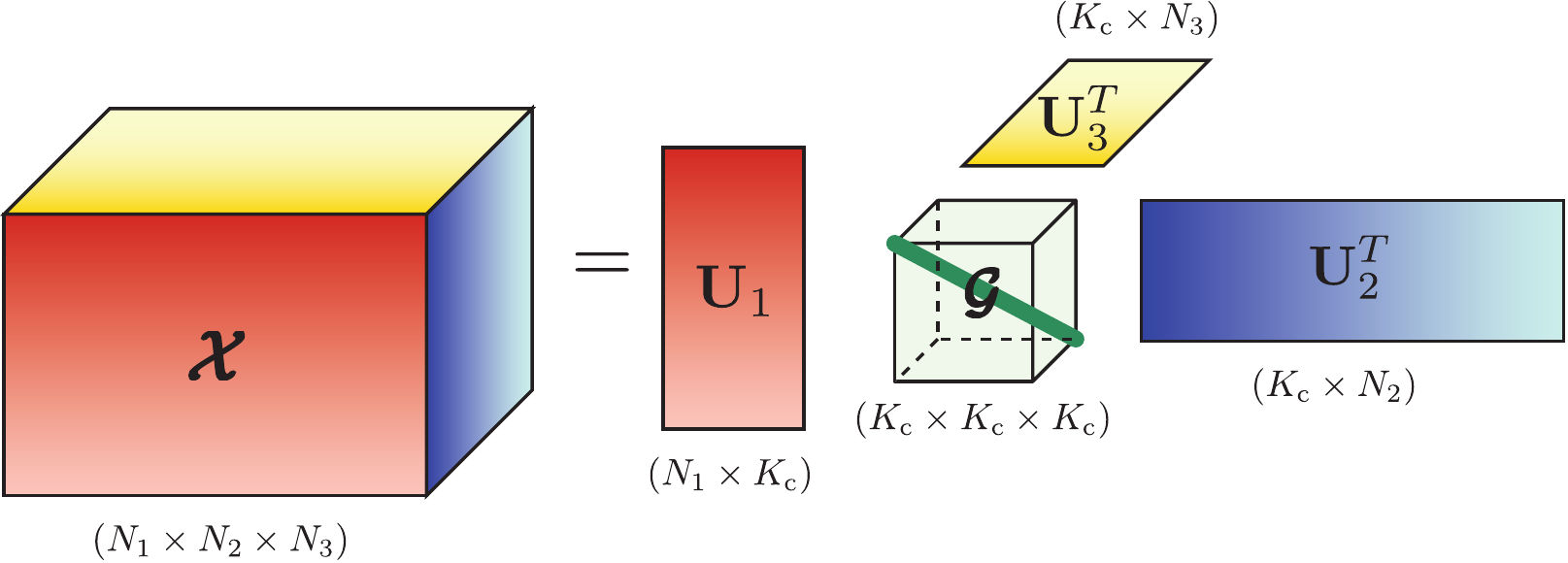}
        \caption{\footnotesize{Diagonal core}}
        \label{fig:diag_core}
    \end{subfigure}
     \caption{\small{Graphic representation of a multilinear system of equations for $R=3$. Colors represent arbitrary values.}}
    \label{fig:tensor_decomposition}
    \vspace{-6mm}
\end{figure}

When the core tensor $\mathbfcal{G}\in\C^{K_{\rm c}\times\dots\times K_{\rm c}}$ is hyperdiagonal (as depicted in Fig.~\ref{fig:diag_core}), \eqref{eq:vec_kron} simplifies to
\begin{equation}
  \vec{x}=\left(\matx{U}_1\odot\dots\odot\matx{U}_R\right){\vec{g}}\label{eq:vec_kr}
\end{equation}
with ${\vec{g}}$ collecting the main diagonal entries of $\mathbfcal{G}$. Note that $\vec{g}$ has different meanings in \eqref{eq:vec_kron} and $\eqref{eq:vec_kr}$, which can always be inferred from the context.

Such a multilinear system is commonly seen with $R=2$ and
 $\mathbfcal{X}= \mathbfcal{G}\bullet_1\matx{U}_1\bullet_2\matx{U}_2 =  \matx{U}_2\mathbfcal{G}\matx{U}_1^T,$
for instance, in image processing when relating an image to its 2-dimensional Fourier transform
with $\mathbfcal{G}$ being the spatial Fourier transform of $\mathbfcal{X}$, and $\matx{U}_1$ and $\matx{U}_2$ being inverse Fourier matrices related to the row and column spaces of the image, respectively. When dealing with Fourier matrices (more generally, Vandermonde matrices) with $\matx{U}_1=\matx{U}_2$ and a diagonal tensor core, $\mathbfcal{X}$ will be a Toeplitz covariance matrix, for which the sampling sets may be designed using sparse covariance sensing~\cite{covariance_sampling_magazine,covariance_sampling}.

%% file: problem_modelling.tex
\section{Problem modeling} \label{sec:modelling}

We are concerned with the design of optimal sampling strategies for an $R$th order tensor signal $\mathbfcal{X}\in\C^{N_1\times\dots\times N_R}$, which admits a multilinear parameterization in terms of a core tensor $\mathbfcal{G}\in\C^{K_1\times\dots\times K_R}$ (dense or diagonal) of smaller dimensionality. We assume that the set of system matrices $\{\matx{U}_i\}_{i=1}^R$ are perfectly known, and that each of them is tall, i.e., $N_i>K_i$ for $i=1,\dots,R$, and has full column rank.

Sparse sampling a tensor $\mathbfcal{X}$ is equivalent to selecting entries of $\vec{x}=\mathrm{vec}(\mathbfcal{X})$. Let $\mathcal{\tilde{N}}$ denote the set of indices of $\vec{x}$. Then, a particular sample selection is determined by a subset of selected indices $\mathcal{L}_\text{un}\subseteq\mathcal{\tilde{N}}$ such that $|\mathcal{L}_\text{un}|=L_\text{un}$ (subscript ``$\mathrm{un}$'' denotes unstructured). This way, we can denote the process of sampling $\mathbfcal{X}$ as a multiplication of $\vec{x}$ by a selection matrix $\matx{\Theta}(\mathcal{L}_\text{un})\in\{0,1\}^{L_\text{un}\times \tilde{N}}$ such that
\begin{equation}
  \vec{y}=\matx{\Theta}(\mathcal{L}_\text{un})\vec{x}= \matx{\Theta}(\mathcal{L}_\text{un})(\matx{U}_1\otimes\dots\otimes\matx{U}_R)\vec{g},\label{eq:sampling_unst_kron}
\end{equation}
for a dense core [cf.~\eqref{eq:vec_kron}], and
\begin{equation}
  \vec{y}=\matx{\Theta}(\mathcal{L}_\text{un})\vec{x}= \matx{\Theta}(\mathcal{L}_\text{un})(\matx{U}_1\odot\dots\odot\matx{U}_R)\vec{g},\label{eq:sampling_unst_kr}
\end{equation}
for a diagonal core [cf.~\eqref{eq:vec_kr}]. Here, $\vec{y}$ is a vector containing the $L_\text{un}$ selected entries of $\vec{x}$ indexed by the set $\mathcal{L}_\text{un}$.

\begin{figure*}[!ht]
  \centering
  \includegraphics[width=\textwidth]{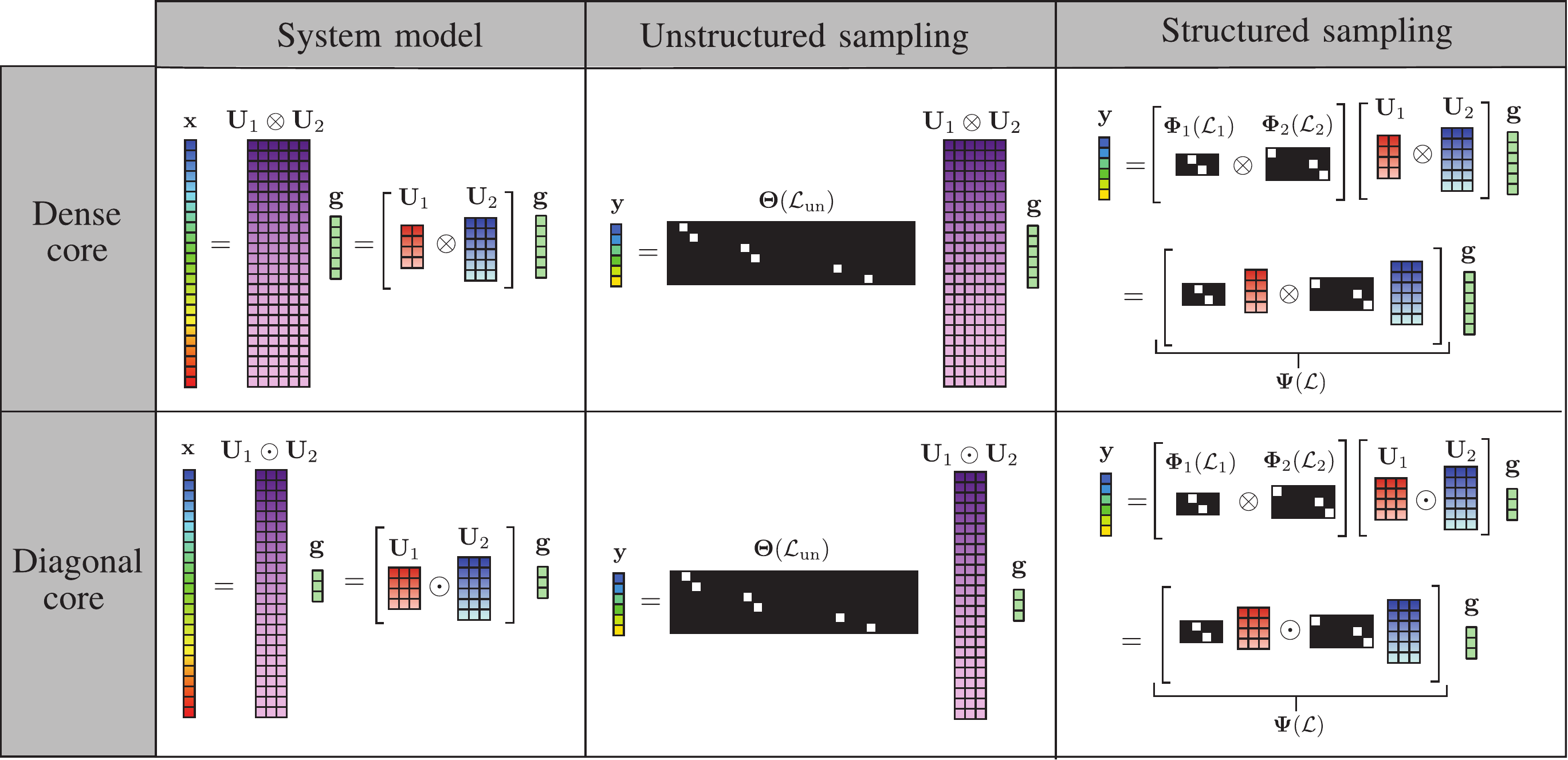}
  \caption{\small{Comparison between unstructured sampling and structured sampling ($R=2$). Black (white) cells represent zero (one) entries, and colored cells represent arbitrary numbers.}}
  \label{fig:sampling}
  \vspace{-5mm}
\end{figure*}

For each case, if $\matx{\Theta}(\mathcal{L}_\text{un})(\matx{U}_1\otimes\dots\otimes\matx{U}_R)$ and $\matx{\Theta}(\mathcal{L}_\text{un})(\matx{U}_1\odot\dots\odot\matx{U}_R)$ have full column rank, then knowing $\vec{y}$ allows to retrieve a unique least squares solution, $\hat{\vec{g}}$, as
\begin{equation}
  \hat{\vec{g}}=\left[\matx{\Theta}(\mathcal{L}_\text{un})(\matx{U}_1\otimes\dots\otimes\matx{U}_R)\right]^\dagger\vec{y},\label{eq:inv_kron_unst}
\end{equation}
or
\begin{equation}
  \hat{\vec{g}}=\left[\matx{\Theta}(\mathcal{L}_\text{un})(\matx{U}_1\odot\dots\odot\matx{U}_R)\right]^\dagger\vec{y}\label{eq:inv_kr_unst},
\end{equation}
depending on whether $\mathbfcal{G}$ is dense or hyperdiagonal. Next, we estimate $\mathbfcal{X}$ using either \eqref{eq:vec_kron} or \eqref{eq:vec_kr}.

In many applications, such as transmitter-receiver placement in multiple input multiple output (MIMO) radar, it is not possible to perform sparse sampling in an unstructured manner by ignoring the underlying domains. For these applications, some unstructured sparse sample selections generally require using a dense sensor selection in each domain (as shown in Fig.~\ref{fig:unst_sparse_sensing}), which produces a significant increase in hardware cost. Also, there is no particular structure in \eqref{eq:inv_kron_unst} and \eqref{eq:inv_kr_unst} that may be exploited to compute the pseudo-inverses, thus leading to a high computational cost to estimate $\vec{x}$. Finally, in the multidomain case, the dimensionality grows rather fast making it difficult to store the matrix $\left(\matx{U}_1\otimes\dots\otimes\matx{U}_R\right)$ or $\left(\matx{U}_1\odot \dots\odot\matx{U}_R\right)$ to perform row subset selection. For all these reasons, we will constrain ourselves to the case where the sampling matrix has a compatible Kronecker structure. In particular, we define a new sampling matrix
\begin{equation}
  \matx{\Phi}(\mathcal{L})\coloneqq\matx{\Phi}_1(\mathcal{L}_1)\otimes\dots\otimes\matx{\Phi}_R(\mathcal{L}_R),\label{eq:separable_phi}
\end{equation}
where each $\matx{\Phi}_i(\mathcal{L}_i)$ represents a selection matrix for the $i$th factor of $\mathbfcal{X}$, $\mathcal{L}_i\subseteq\mathcal{N}_i$ is the set of selected row indices from the matrix $\matx{U}_i$ for $i=1,\dots,R$, and 
$
  \mathcal{L}=\bigcup_{i=1}^R\mathcal{L}_i$ and $\mathcal{L}_i\cap\mathcal{L}_j=\varnothing\;\text{for } i\neq j.
$

We will use the notation $|\mathcal{L}_i|=L_i$ and $|\mathcal{L}|=\sum_{i=1}^R L_i=L$ to denote the number of selected sensors per domain and the total number of selected sensors, respectively; whereas $\mathcal{\tilde{L}}=\mathcal{L}_1\times\dots\times\mathcal{L}_R$ and $\tilde{L}=|\mathcal{\tilde{L}}|=\prod_{i=1}^R L_i$ denote the set of sample indices and the total number of samples acquired with the above Kronecker-structured sampler. In order to simplify the notation, whenever it will be clear, we will drop the explicit dependency of $\matx{\Phi}_i(\mathcal{L}_i)$ on the set of selected rows $\mathcal{L}_i$, from now on, and simply use $\matx{\Phi}_i$.

Imposing a Kronecker structure on the sampling scheme means that sampling can be performed independently for each domain. In the \emph{dense core tensor} case [cf.~\eqref{eq:vec_kron}], we have 
\begin{align*}
  \vec{y}&=\left(\matx{\Phi}_1\otimes\dots\otimes\matx{\Phi}_R\right)\left(\matx{U}_1\otimes\dots\otimes\matx{U}_R\right)\vec{g}\\
  &=\left(\matx{\Phi}_1\matx{U}_1\otimes\dots\otimes\matx{\Phi}_R\matx{U}_R\right)\vec{g}= \matx{\Psi}(\mathcal{L})\vec{g},\numberthis\label{eq:sampling_eq_kron}
\end{align*}
whereas in the \emph{diagonal core tensor} case [cf.~\eqref{eq:vec_kr}], we have
\begin{align*}
  \vec{y}&=\left(\matx{\Phi}_1\otimes\dots\otimes\matx{\Phi}_R\right)\left(\matx{U}_1\odot\dots\odot\matx{U}_R\right){\vec{g}}\\
  &=\left(\matx{\Phi}_1\matx{U}_1\odot\dots\odot\matx{\Phi}_R\matx{U}_R\right){\vec{g}}=\matx{\Psi}(\mathcal{L}){\vec{g}}.\numberthis\label{eq:sampling_eq_kr}
\end{align*}

As in the unstructured case, whenever \eqref{eq:sampling_eq_kron} or \eqref{eq:sampling_eq_kr} are overdetermined, using least squares, we can estimate the core $\vec{\hat{g}} = \matx{\Psi}^\dagger(\mathcal{L})\vec{y} $ as
\begin{equation}
  \vec{\hat{g}}=\left[\left(\matx{\Phi}_1\matx{U}_1\right)^\dagger\otimes\dots\otimes\left(\matx{\Phi}_R\matx{U}_R\right)^\dagger\right]\vec{y}\numberthis\label{eq:inv_kron},
\end{equation}
or
\begin{align*}
  \vec{\hat{g}}&=\left[\left(\matx{\Phi}_1\matx{U}_1\right)^H\left(\matx{\Phi}_1\matx{U}_1\right)\circ\dots\circ\left(\matx{\Phi}_R\matx{U}_R\right)^H\left(\matx{\Phi}_R\matx{U}_R\right)\right]^\dagger\\
  &\times\left[\left(\matx{\Phi}_1\matx{U}_1\right)^H\odot\dots\odot\left(\matx{\Phi}_R\matx{U}_R\right)^H\right]\vec{y}\numberthis\label{eq:inv_kr},
  \end{align*}
and then reconstruct $\vec{\hat{x}}$ using \eqref{eq:vec_kron} or \eqref{eq:vec_kr}, respectively. Comparing \eqref{eq:inv_kron} and \eqref{eq:inv_kr} to \eqref{eq:inv_kron_unst} and \eqref{eq:inv_kr_unst} we can see that leveraging the Kronecker structure of the proposed sampling scheme allows to greatly reduce the computational complexity of the least-squares problem, as the pseudoinverses in \eqref{eq:inv_kron} and \eqref{eq:inv_kr} are taken on matrices of a much smaller dimensionality than in \eqref{eq:inv_kron_unst} and \eqref{eq:inv_kr_unst}. An illustration of the comparison between unstructured sparse sensing and Kronecker-structured sparse sensing is shown in Fig.~\ref{fig:sampling} for $R=2$.

Suppose the measurements collected in $\vec{y}$ are perturbed by zero-mean white Gaussian noise with unit variance, then the least-squares solution has the inverse error covariance or the Fisher information matrix $\matx{T}(\mathcal{L})=\mathbb{E}\{({\vec{g}}-\hat{\vec{g}})(\vec{g}-\hat{\vec{g}})^H\} = \matx{\Psi}^H(\mathcal{L})\matx{\Psi}(\mathcal{L})$ that determines the quality of the estimators $\hat{\vec{g}}$. Therefore, we can use scalar functions of $\matx{T}(\mathcal{L})$ as a figure of merit to propose the sparse tensor sampling problem 
\begin{equation}
\label{eq:sparsesensing}
  \optimize_{\mathcal{L}_1, \dots,\mathcal{L}_R}  f\left\{\matx{T}(\mathcal{L})\right\} \,\, 
  \text{s. to} \,\,  \sum_{i=1}^R|\mathcal{L}_i|=L, \, \mathcal{L}=\bigcup_{i=1}^R\mathcal{L}_i,
\end{equation}
where with ``optimize'' we mean either ``maximize'' or ``minimize'' depending on the choice of the scalar function $f\{\cdot\}$. Solving \eqref{eq:sparsesensing} is not trivial due to the cardinality constraints. Therefore, in the following, we will propose tight surrogates for typically used scalar performance metrics $f\{\cdot\}$ in design of experiments with which the above discrete optimization problem can be solved efficiently and near optimally.

Note that the cardinality constraint in \eqref{eq:sparsesensing} restricts the total number of selected sensors to $L$, without imposing any constraint on the total number of gathered samples $\tilde{L}$. Although the maximum number of samples can be constrained using the constraint $\sum_{i=1}^R\log|\mathcal{L}_i|\leq \tilde{L}$, the resulting near-optimal solvers are computationally intense with a complexity of about $\mathcal{O}(N^5)$ \cite{sub_knapsack,sub_knapsack2}. Such heuristics are not suitable for the large-scale scenarios of interest.

\subsection{Prior art}

Choosing the best subset of measurements from a large set of candidate sensing locations has received a lot of attention, particularly for $R=1$, usually under the name of sensor selection/placement, which also is more generally referred to as sparse sensing~\cite{foundations}. 

Typically sparse sensing design is posed as a discrete optimization problem that finds the best sampling subset by optimizing a scalar function of the error covariance matrix. Some of the popular choices are, to minimize the mean squared error (MSE): $f\left\{\matx{T}(\mathcal{L})\right\} := \tr{\matx{T}^{-1}}$ or the frame potential: $f\left\{\matx{T}(\mathcal{L})\right\}:=\tr{\matx{T}^H\matx{T}}$, or to maximize $f\left\{\matx{T}(\mathcal{L})\right\}:= \lambda_{\text{min}}\{\matx{T}\}$ or $f\left\{\matx{T}(\mathcal{L})\right\} := \log \det{\matx{T}}$.
In this work, we will focus on the frame potential criterium as we will show later that this metric leads to very efficient sampler designs.

Depending on the strategy used to solve the optimization problem  \eqref{eq:sparsesensing} we can classify the prior art in two categories: solvers based on convex optimization, and greedy methods that leverage submodularity. In the former category, \cite{cvx_sampling} and \cite{nonlinear} present several convex relaxations of the sparse sensing problem for different optimality criteria for inverse problems with linear and non-linear models, respectively. In particular, due to the Boolean nature of the sensor selection problem (i.e., a sensor is either selected or not), its related optimization problem is not convex. However, these constraints and the constraint on the number of selected sensors can be relaxed, and once the relaxed convex problem is solved, a thresholding heuristic (deterministic or randomized) can be used to recover a Boolean solution. Despite its good performance, the complexity of convex optimization solvers is rather high (cubic with the dimensionality of the signal). Therefore, the use of convex optimization approaches to solve the sparse sensing problem in large-scale scenarios, such as the sparse tensor sampling problem, gets even more computationally intense.

For high-dimensional scenarios, greedy methods (algorithms that select one sensor at a time) are more useful. A greedy algorithm scales linearly with the number of sensors, and if one can prove submodularity of the objective function, its solution has a multiplicative near-optimality guarantee~\cite{submodular_1}. Several authors have followed this strategy and have proved submodularity of different optimality criteria such as \emph{D}-optimality \cite{logdet}, mutual information \cite{krause2008near}, and frame potential \cite{frame_potential}. All of them for the case $R=1$.

Besides parameter estimation, sparse sensing has also been studied for other common signal processing tasks, like detection \cite{chepuri2016sparse,yu1997sampling} and filtering \cite{chepuri2014sparsity,masazade2012sparsity}.  In a different context, the extension of Compressed Sensing (CS) to multidomain signals has been extensively studied \cite{kron_cs,kron_omp,caiafa2013multidimensional,mri_cs}. CS is many times seen as a complementary sampling framework to sparse sensing~\cite{foundations}, wherein CS the focus is on recovering a sparse signal rather than on designing a sparse measurement space. 

\subsection{Our contributions}

In this paper, we extend the sparse sampling framework to \emph{multilinear inverse problems}. We refer to it as ``sparse tensor sampling". We focus on two particular cases, depending on the structure of the core tensor $\mathbfcal{G}$:

\begin{itemize}
  \item \emph{Dense core}: Whenever the core tensor is non-diagonal,  sampling is performed based on \eqref{eq:sampling_eq_kron}. We will see that to ensure identifiability of the system, we need to select more entries in each domain than the rank of its corresponding system matrix, i.e., as a necessary condition we require $L\geq\sum_{i=1}^R K_i = K$ sensors, where $\{K_i\}_{i=1}^R$ are the dimensions of the core tensor $\mathbfcal{G}$. 
  \item \emph{Diagonal core}: Whenever the core tensor is diagonal, sampling is performed based on~\eqref{eq:sampling_eq_kr}. The use of the Khatri-Rao product allows for higher compression. In particular, under mild conditions on the entries of the factor matrices, we can guarantee identifiability of the sampled equations using $L\geq K_\text{c}+R-1$ sensors, where $K_c$ is the length of the edges of the hypercubic core $\mathbfcal{G}$. 
\end{itemize}
For both the cases, we propose efficient greedy algorithms to compute a near-optimal sampling set.

\subsection{Paper outline}

The remainder of the paper is organized as follows. In Sections~\ref{sec:kron_sampling} and \ref{sec:diag_sampling}, we develop solvers for the sparse tensor sampling problem with dense and diagonal core tensors, respectively. In Section~\ref{sec:numerical_results}, we provide a few examples to illustrate the developed framework. Finally, we conclude this paper by  summarizing the results in Section~\ref{sec:conclusions}.

%% file: dense_sampling.tex
\section{Dense core sampling}\label{sec:kron_sampling}

In this section, we focus on the most general situation when $\mathbfcal{G}$ is an unstructured dense tensor. Our objective is to design the sampling sets $\{\mathcal{L}_i\}_{i=1}^R$ by solving the discrete optimization problem~\eqref{eq:sparsesensing}.  

We formulate the sparse tensor sampling problem using the frame potential as a performance measure. Following the same rationale as in \cite{frame_potential}, but for multidomain signals, we will argue that the frame potential is a tight surrogate of the MSE. By doing so, we will see that when we impose a Kronecker structure on the sampling scheme, as in \eqref{eq:separable_phi}, the frame potential of $\matx{\Psi}$ can be factorized in terms of the frame potential of the different domain factors. This allows us to propose a low complexity algorithm for sampling tensor data. 

Throughout this section, we will use tools from submodular optimization theory. Hence, we will start by introducing the main concepts related to submodularity in the next subsection. 
\subsection{Submodular optimization}
Submodularity is a notion based on the law of diminishing returns \cite{submodular_book} that is useful to obtain heuristic algorithms with near-optimality guarantees for cardinality-constrained discrete optimization problems. More precisely, submodularity is formally defined as follows.

\begin{definition}[Submodular function \cite{submodular_book}]\label{def:submodular}
  A set function $f:2^\mathcal{N}\rightarrow\R$ defined over the subsets of $\mathcal{N}$ is submodular if, for every $\mathcal{X}\subseteq\mathcal{N}$ and $x,y\in\mathcal{N}\setminus\mathcal{X}$, we have
  \begin{equation*}
    f(\mathcal{X}\cup\{x\})-f(\mathcal{X})\geq f(\mathcal{X}\cup\{x,y\})-f(\mathcal{X}\cup\{y\}).
  \end{equation*}
  A function $f$ is said to be supermodular if $-f$ is submodular.
\end{definition}

Besides submodularity, many near-optimality theorems in discrete optimization require functions to be also monotone non-decreasing, and normalized.

\begin{definition}[Monotonicity]
  A set function $f:2^\mathcal{N}\rightarrow\R$ is monotone non-decreasing if, for every $\mathcal{X}\subseteq\mathcal{N}$,
  \begin{equation*}
    f(\mathcal{X}\cup\{x\})\geq f(\mathcal{X})\quad \forall x\in\mathcal{N}\setminus\mathcal{X}
  \end{equation*}
\end{definition}

\begin{definition}[Normalization]
  A set function $f:2^\mathcal{N}\rightarrow\R$ is normalized if $f(\varnothing)=0$.
\end{definition}

In submodular optimization, matroids are generally used to impose constraints on an optimization, such as the ones in~\eqref{eq:sparsesensing}. A matroid generalizes the concept of linear independence in algebra to sets. Formally, a matroid is defined as follows.

\begin{definition}[Matroid \cite{discrete_optimization}]\label{def:matroid}
  A finite matroid $\mathcal{M}$ is a pair $(\mathcal{N},\mathcal{I})$, where $\mathcal{N}$ is a finite set and $\mathcal{I}$ is a family of subsets of $\mathcal{N}$ that satisfies: 1) The empty set is independent, i.e., $\varnothing\in\mathcal{I}$; 2) For every $\mathcal{X}\subseteq\mathcal{Y}\subseteq\mathcal{N}$, if $\mathcal{Y}\in\mathcal{I}$, then $\mathcal{X}\in\mathcal{I}$; and 3) For every $\mathcal{X},\mathcal{Y}\subseteq\mathcal{N}$ with $|\mathcal{Y}|>|\mathcal{X}|$ and $\mathcal{X},\mathcal{Y}\in\mathcal{I}$ there exists one $x\in\mathcal{Y}\setminus\mathcal{X}$ such that $\mathcal{X}\cup\{x\}\in\mathcal{I}$.
\end{definition}

In this paper we will deal with the following types of matroids.
\begin{example}[Uniform matroid \cite{discrete_optimization}]
  The subsets of $\mathcal{N}$ with at most $K$ elements form a uniform matroid $\mathcal{M}_\mathrm{u}=(\mathcal{N},\mathcal{I}_\mathrm{u})$ with $\mathcal{I}_\mathrm{u}=\{\mathcal{X}\subseteq\mathcal{N}:|\mathcal{X}|\leq K\}$.
\end{example}

\begin{example}[Partition matroid \cite{discrete_optimization}]\label{ex:partition_matroid}
  If $\{\mathcal{N}_i\}_{i=1}^R$ form a partition of $\mathcal{N}=\bigcup_{i=1}^R\mathcal{N}_i$ then $\mathcal{M}_\mathrm{p}=(\mathcal{N},\mathcal{I}_\mathrm{p})$ with $\mathcal{I}_\mathrm{p}=\{\mathcal{X}\subseteq\mathcal{N}:|\mathcal{X}\cap\mathcal{N}_i|\leq K_i\quad i=1,\dots,R\}$ defines a partition matroid.
\end{example}

\begin{example}[Truncated partition matroid \cite{discrete_optimization}]
  The intersection of a uniform matroid $\mathcal{M}_\mathrm{u}=(\mathcal{N},\mathcal{I}_\mathrm{u})$ and a partition matroid $\mathcal{M}_\mathrm{p}=(\mathcal{N},\mathcal{I}_\mathrm{p})$ defines a truncated partition matroid $\mathcal{M}_\mathrm{t}=(\mathcal{N},\mathcal{I}_\mathrm{p}\cap\mathcal{I}_\mathrm{u})$.
\end{example}

The matroid-constrained submodular optimization problem
\begin{equation}
  \maximize_{\mathcal{X}\subseteq\mathcal{N}} f(\mathcal{X})\quad \text{subject to}\quad \mathcal{X}\in\bigcap_{i=1}^T\mathcal{I}_i\label{eq:submodular_matroid}
\end{equation}
can be solved near optimally using Algorithm~\ref{alg:greedy_matroid}. This result is formally stated in the following theorem.

\begin{algorithm}[t]
\begin{algorithmic}[1]
\Require{$\mathcal{X}=\varnothing$, $K$, $\{\mathcal{I}_i\}_{i=1}^T$}
\For{$k\gets 1$ \textbf{to} $K$}
\State $s^\star=\argmax_{s\notin\mathcal{X}}\{f(\mathcal{X}\cup\{s\}):\mathcal{X}\cup\{s\}\in\bigcap_{i=1}^T\mathcal{I}_i\}$ 
\State $\mathcal{X}\leftarrow\mathcal{X}\cup\{s^\star\}$
\EndFor
\State\Return $\mathcal{X}$
\end{algorithmic}
\caption{Greedy maximization subject to $T$ matroid constraints}
\label{alg:greedy_matroid}
\end{algorithm}

\begin{theorem}[Matroid-constrained submodular maximization\cite{submodular_2}]\label{thm:submodular_matroid}
  Let $f:2^\mathcal{N}\rightarrow\R$ be a monotone non-decreasing, normalized, submodular set function, and $\{\mathcal{M}_i=(\mathcal{N},\mathcal{I}_i)\}_{i=1}^T$ be a set of matroids defined over $\mathcal{N}$. Furthermore, let $\mathcal{X}^\star$ denote the optimal solution of \eqref{eq:submodular_matroid}, and let $\mathcal{X}_\text{greedy}$ be the solution obtained by Algorithm~\ref{alg:greedy_matroid}. Then
  \begin{equation*}
    f(\mathcal{X}_\text{greedy})\geq \cfrac{1}{T+1}\;f(\mathcal{X}^\star).
  \end{equation*}
\end{theorem}

\subsection{Greedy method}\label{sec:dense_fp}

The frame potential \cite{mse_fp} of the matrix $\matx{\Psi}$ is defined as the trace of the Grammian matrix
$
  \FP{\matx{\Psi}}\coloneqq\tr{\matx{T}^H\matx{T}}
$  
with $\matx{T}=\matx{\Psi}^H\matx{\Psi}$. The frame potential can be related to the MSE, 
$
  \MSE(\matx{\Psi}(\mathcal{L}))=\tr{\matx{T}^{-1}(\mathcal{L})},
 $ 
using~\cite{frame_potential} 
\begin{equation}
  c_1\cfrac{\FP{\matx{\Psi}(\mathcal{L})}}{\lambda_\mathrm{max}^2\{\matx{T}(\mathcal{L})\}}\leq\MSE(\matx{\Psi}(\mathcal{L}))\leq c_2\cfrac{\FP{\matx{\Psi}(\mathcal{L})}}{\lambda_\mathrm{min}^2\{\matx{T}(\mathcal{L})\}},\label{eq:MSE_bound}
\end{equation}
where $c_1$, and $c_2$ are constants that depend the data model.

From the above bound, it is clear that by minimizing the frame potential of $\matx{\Psi}$ one can minimize the MSE, which is otherwise difficult to minimize as it is neither convex, nor submodular.

The frame potential of $\matx{\Psi}(\mathcal{L})\coloneqq\matx{\Psi}_1(\mathcal{L}_1)\otimes\dots\otimes\matx{\Psi}_R(\mathcal{L}_R)$  can be expressed as the frame potential of its factors 
$\matx{\Psi}_i(\mathcal{L}_i)\coloneqq\matx{\Phi}_i(\mathcal{L}_i)\matx{U}_i$. To show this, recall the definition of the frame potential as
\begin{align*}
  \FP{\matx{\Psi}(\mathcal{L})}&=\tr{\matx{T}^H(\mathcal{L})\matx{T}(\mathcal{L})}\\
  &=\tr{\matx{T}_1^H\matx{T}_1\otimes\dots\otimes\matx{T}_R^H\matx{T}_R},\numberthis\label{eq:kron_fp_incomplete}
\end{align*}
where $\matx{T}_i=\matx{\Psi}_i^H\matx{\Psi}_i$.
Now, using the fact that for any two matrices $\matx{A}\in\C^{K_A\times K_A}$ and $\matx{B}\in\C^{K_B\times K_B}$ we have $\tr{\matx{A}\otimes\matx{B}}=\tr{\matx{A}}\tr{\matx{B}}$, we can expand \eqref{eq:kron_fp_incomplete} as
\begin{equation*}
  \FP{\matx{\Psi}(\mathcal{L})}=\prod_{i=1}^R\tr{\matx{T}_i^H\matx{T}_i}=\prod_{i=1}^R \FP{\matx{\Psi}_i(\mathcal{L}_i)}.
\end{equation*}
For brevity, we will write the above expression alternatively as an explicit function of the selection sets $\mathcal{L}_i$:
\begin{align}
  F(\mathcal{L})&\coloneqq\FP{\matx{\Psi}(\mathcal{L})}=\prod_{i=1}^RF_i(\mathcal{L}_i) \coloneqq \prod_{i=1}^R \FP{\matx{\Psi}_i(\mathcal{L}_i)}\label{eq:kron_fp}
\end{align}
Expression \eqref{eq:kron_fp} shows again the advantage of working with a Kronecker-structured sampler: instead of computing every cross-product between the columns of $\matx{\Psi}$ to compute the frame potential, we can arrive to the same value using the frame potential of $\{\matx{\Psi}_i\}_{i=1}^R$.

\subsubsection{Submodularity of $F(\mathcal{L})$}

Function $F(\mathcal{L})$ as defined in \eqref{eq:kron_fp} does not directly meet the conditions [cf. Theorem~\ref{thm:submodular_matroid}] required for near optimality of the greedy heuristic, but it can be modified slightly to satisfy them. In this sense, we define the function $G:2^\mathcal{N}\rightarrow\R$ on the subsets of $\mathcal{N}$ as
\begin{equation}
\label{eq:set_func}
  G(\mathcal{S})\coloneqq F(\mathcal{N})- F(\mathcal{N}\setminus\mathcal{S})
\end{equation}
where recall that $F(\mathcal{N}) = \prod_{i=1}^R F_i(\mathcal{N}_i)$, $F(\mathcal{N}\setminus\mathcal{S}) = \prod_{i=1}^RF_i(\mathcal{N}_i\setminus\mathcal{S}_i)$, and 
$
  \mathcal{S}=\bigcup_{i=1}^R\mathcal{S}_i, \quad \mathcal{S}_i\cap\mathcal{S}_j=\varnothing\text{ for }i\neq j$. Therefore, $\{\mathcal{S}_i\}_{i=1}^R$ form a partition of $\mathcal{S}$. 

It is clear that if we make the change of variables from $\mathcal{L}$ to $\mathcal{S}$ maximizing $G$ over $\mathcal{S}$ is the same as minimizing the frame potential over $\mathcal{L}$. However, working with the complement set results in a set function that is submodular and monotone non-decreasing, as shown in the next theorem. Consequently, $G$ satisfies the conditions of the near-optimality theorems.
\begin{theorem}\label{thm:submodular_G}
  The set function $G(\mathcal{S})$ defined in \eqref{eq:set_func} is a normalized, monotone non-decreasing, submodular function for all subsets of $\mathcal{N}=\bigcup_{i=1}^R\mathcal{N}_i$. 
\end{theorem}
\begin{proof} 
See Appendix~\ref{ap:proof_G}.
\end{proof}

With this result we can now claim near-optimality of the greedy algorithm that solves the cardinality constrained maximization of $G(\mathcal{S})$. However, as we said, minimizing the frame potential only makes sense as long as \eqref{eq:MSE_bound} is tight. In particular, whenever $\matx{T}(\mathcal{L})$ is singular we know that the MSE is infinity, and hence \eqref{eq:MSE_bound} is meaningless. For this reason, next to the cardinality constraint in \eqref{eq:sparsesensing} that limits the total number of sensors, we need to ensure that $\matx{\Psi}(\mathcal{L})$ has full column rank, i.e., $L_i\geq K_i$ for $i=1,\dots,R$. In terms of $\mathcal{S}$, this is equivalent to
\begin{equation}
  |\mathcal{S}_i|=|\mathcal{N}_i\setminus\mathcal{L}_i|\leq N_i-K_i\quad i=1,\dots,R,\label{eq:identifiability_constraint}
\end{equation}
where this set of constraints forms a partition matroid $\mathcal{M}_p=(\mathcal{N},\mathcal{I}_p)$ [cf. Example~\ref{ex:partition_matroid} from Definition~\ref{def:matroid}]. Hence, we can introduce the following submodular optimization problem as surrogate for the minimization of the frame potential
   \begin{equation}
   \label{eq:submodular_max_kron_matroid}
    \maximize_{\mathcal{S}\subseteq\mathcal{N}} \, G(\mathcal{S}) \,\,\, \text{s. to} \,\,\,  \mathcal{S} \in \mathcal{I}_\mathrm{u}\cap\mathcal{I}_p
   \end{equation} 
with $\mathcal{I}_\mathrm{u}=\{\mathcal{A}\subseteq\mathcal{N}:|\mathcal{A}|\leq N-L\}$ and $\mathcal{I}_p=\{\mathcal{A}\subseteq\mathcal{N}:|\mathcal{A}\cap\mathcal{N}_i|\leq N_i-K_i\;\;i=1,\dots,R\}$.
Theorem~\ref{thm:submodular_matroid} gives, therefore, all the ingredients to assess the near-optimality of Algorithm~\ref{alg:greedy_matroid} applied on \eqref{eq:submodular_max_kron_matroid}, for which the results are particularized as the following corollary.

\begin{corollary}\label{corollary:bound_G}
  The greedy solution $\mathcal{S}_\text{greedy}$ to \eqref{eq:submodular_max_kron_matroid} obtained from Algorithm~\ref{alg:greedy_matroid}) is $1/2$-near-optimal, i.e.,
$G(\mathcal{S}_\text{greedy})\geq \frac{1}{2}G(\mathcal{S}^\star).$
\end{corollary}

\begin{proof}
  Follows from Theorem~\ref{thm:submodular_matroid}, and since \eqref{eq:submodular_max_kron_matroid} has $T=1$ (truncated-partition) matroid constraint. 
\end{proof}

Next, we compute an explicit bound with respect to the frame potential of $\matx{\Psi}$, which is the objective function we initially wanted to minimize. This bound is given in the following theorem.

\begin{theorem}\label{thm:approx_fp}
  The greedy solution $\mathcal{L}_\text{greedy}$ to \eqref{eq:submodular_max_kron_matroid} obtained from Algorithm~\ref{alg:greedy_matroid} is near optimal with respect to the frame potential as
$ F(\mathcal{L}_\text{greedy})\leq \gamma F(\mathcal{L}^\star)$
  with $\gamma=\cfrac{1}{2}\left(\cfrac{K}{L_{\text{min}}^2}\prod_{i=1}^R F_i(\mathcal{N}_i)+1\right)$, and $L_\text{min}=\min_{i\in\mathcal{L}}\norm{\vec{u}_i}^2_2$, being $\vec{u}_i$ the $i$th row of $(\matx{U}_1\otimes\dots\otimes\matx{U}_R)$.
\end{theorem}
\begin{proof}
Obtained similar to the bound in \cite{frame_potential}, but specialized for \eqref{eq:kron_fp} and 1/2-near-optimality; see~\cite{thesis} for details.    
\end{proof}

As with the $R=1$ case in~\cite{frame_potential}, $\gamma$ is heavily influenced by the frame potential of $(\matx{U}_1\otimes\dots\otimes\matx{U}_R)$. Specifically, approximation gets tighter when $F(\mathcal{N})$ is small or the core tensor dimensionality decreases.

\subsubsection{Computational complexity} 

The running time of Algorithm~\ref{alg:greedy_matroid} applied to solve \eqref{eq:submodular_max_kron_matroid} can greatly be reduced by precomputing the inner products between the rows of every $\matx{U}_i$ before starting the iterations. This has a complexity of $\mathcal{O}(N_i^2K_i)$ for each domain. Once these inner products are computed, in each iteration we need to find $R$ times the maximum over $\mathcal{O}(N_i)$ elements. Since we run $N-L$ iterations, the complexity of all iterations is $\mathcal{O}(N^2_\text{max})$, with $N_\text{max}=\max_i N_i$. Therefore, the total computational complexity of the greedy method is $\mathcal{O}(N_\text{max}^2K_\text{max})$ with $K_\text{max}=\max_i K_i$.

\subsubsection{Practical considerations}\label{sec:practical_fp}

Due to the characteristics of the greedy iterations, the algorithm tends to give solutions with a very unbalanced cardinality. In particular, for most situations, the algorithm chooses one of the domains in the first few iterations and empties that set till it hits the identifiability constraint of that domain. Then, it proceeds to another domain and empties it as well, and so on.  This is due to the objective function, which is a product of smaller objectives. Indeed, if we are asked to minimize a product of two elements by subtracting a value from them, it is generally better to subtract from the smallest element. Hence, if this minimization is performed multiple times we will tend to remove always from the same element.

The consequences of this behavior are twofold. On the one hand, this greedy method tends to give a sensor placement that yields a very small number of samples $\tilde{L}$, as we will also see in the simulations. Therefore, when comparing this method to other sensor selection schemes that produce solutions with a larger $\tilde{L}$ it generally ranks worse in MSE for a given $L$. On the other hand, the solution of this scheme tends to be tight on the identifiability constraints for most of the domains, thus hampering the performance on those domains. This implication, however, has a simple solution. By introducing a small slack variable $\alpha_i>0$ to the constraints, we can obtain a sensor selection which is not tight on the constraints. This amounts to solving the problem
\begin{align*}
  &\maximize_{\mathcal{S}\subseteq\mathcal{N}}  G(\mathcal{S})\numberthis\\
  &\text{s. to } \quad |\mathcal{S}|= N-L, 
   |\mathcal{S}\cap\mathcal{N}_i|\leq N_i-K_i-\alpha_i,  i=1,\dots,R.
\end{align*}
Tuning $\{\alpha_i\}_{i=1}^R$ allows to regularize the tradeoff between compression and accuracy of the greedy solution.

We conclude this section with a remark on an alternative performance measure. 
\begin{remark} As a alternative performance measure, one can think on maximizing the set function $\log \det{\matx{T}(\mathcal{L})}$. Although this set function can be shown to be submodular over all subsets of $\cal{N}$~\cite{thesis}, the related greedy algorithm cannot be  constrained to always result in an identifiable system after subsampling. Thus, its applicability is more limited than the frame potential formulation; see~\cite{thesis} for a more detailed discussion. 
\end{remark}

%% file: diagonal_sampling.tex
\section{Diagonal core sampling}\label{sec:diag_sampling}

So far, we have focused on the case when $\mathbfcal{G}$ is dense and has no particular structure. In that case, we have seen that we require at least $\sum_{i=1}^R K_i$ sensors to recover our signal with a finite MSE. In many cases of interest, $\mathbfcal{G}$ admits a structure. In particular, in this section, we investigate the case when $\mathbfcal{G}$ is a diagonal tensor. Under some mild conditions on the entries of $\{\matx{U}\}_{i=1}^R$, we can leverage the structure of $\mathbfcal{G}$ to further increase the compression. As before, we  develop an efficient and near-optimal greedy algorithm based on minimizing the frame potential to design the sampling set.

\subsection{Identifiability conditions} \label{sec:identifiability_constraints}

In contrast to the dense core case, the number of unknowns in a multilinear system of equations with a diagonal core does not increase with the tensor order, whereas for a dense core it grows exponentially. This means that when sampling signals with a diagonal core decomposition, one can expect a stronger compression. 

To derive the identifiability conditions for \eqref{eq:sampling_eq_kr}, we present the result from~\cite{rank_kr} as the following theorem.
\begin{theorem}[Rank of Khatri-Rao product \cite{rank_kr}]\label{thm:rank_rk}
Let $\matx{A}\in\C^{N\times K}$ and $\matx{B}\in\C^{M\times K}$ be two matrices with no all-zero column. Then,
  \begin{equation*}
    \mathrm{rank}(\matx{A}\odot\matx{B})\geq\max\{\mathrm{rank}(\matx{A}),\mathrm{rank}(\matx{B})\}.
  \end{equation*}
\end{theorem}

Based on the above theorem, we can give the following sufficient conditions for identifiability of the system \eqref{eq:sampling_eq_kr}. 
\begin{corollary}
Let $z_i$ denote the maximum number of zero entries in any column of $\matx{U}_i$. If for every $\matx{\Psi}_i(\mathcal{L}_i)$ we have $|\mathcal{L}_i|>z_i$, and there is at least one $\matx{\Psi}_j$ with $\mathrm{rank}(\matx{\Psi}_j)=K_\text{c}$, then $\matx{\Psi}(\mathcal{L})$ has full column rank.
\end{corollary}
\begin{proof}
Selecting $L_i>z_i$ rows from each $\matx{U}_i$ ensures that no $\matx{\Psi}_i$ will have an all-zero column. Then, if for at least one $\matx{\Psi}_j$ we have $\mathrm{rank}(\matx{\Psi}_j)=K_\text{c}$, then due to Theorem~\ref{thm:rank_rk} we have
  \begin{align*}
    \mathrm{rank}(\matx{\Psi}(\mathcal{L}))&\geq\max_{i=1,\dots,R}\{\mathrm{rank}(\matx{\Psi}_i)\}\\
    &=\max\left\{\mathrm{rank}(\matx{\Psi}_j),\max_{i\neq j}\{\mathrm{rank}(\matx{\Psi}_j)\}\right\}=K_\text{c}.
  \end{align*}
\end{proof}
Therefore, in order to guarantee identifiability we need to select $L_j\geq \max\{K_\text{c}, z_j+1\}$ rows from any factor matrix $j$, and $L_i\geq\max\{1, z_i+1\}$ from the other factors with $i\neq j$. In many scenarios, we usually have $\{z_i=0\}_{i=1}^R$ since no entry in $\{\matx{U}_i\}_{i=1}^R$ will exactly be zero. In those situations we will require to select at least $L=\sum_{i=1}^R L_i\geq K_\text{c} + R -1$ elements.

\subsection{Greedy method}

As we did for the case with a dense core, we start by finding an expression for the frame potential of a Khatri-Rao product in terms of its factors. The Grammian matrix $\matx{T}(\mathcal{L})$ of a diagonal core tensor decomposition has the form
\begin{align*}
  \matx{T}&=\matx{\Psi}^H\matx{\Psi}=\left(\matx{\Psi}_1\odot\dots\odot\matx{\Psi}_R\right)^H\left(\matx{\Psi}_1\odot\dots\odot\matx{\Psi}_R\right)\\
  &=\matx{\Psi}_1^H\matx{\Psi}_1\circ\dots\circ\matx{\Psi}_R^H\matx{\Psi}_R=\matx{T}_1\circ\dots\circ\matx{T}_R.
\end{align*}
Using this expression, the frame potential of a Khatri-Rao product becomes
\begin{equation}
  \FP{\matx{\Psi}}=\tr{\matx{T}^H\matx{T}}=\norm{\matx{T}}_F^2=\norm{\matx{T}_1\circ\dots\circ\matx{T}_R}_F^2.\label{eq:kr_fp}
\end{equation}
For brevity, we will denote the frame potential as an explicit function of the selected set as
\begin{equation}
  P(\mathcal{L})\coloneqq\FP{\matx{\Psi}(\mathcal{L})}=\norm{\matx{T}_1(\mathcal{L}_1)\circ\dots\circ\matx{T}_R(\mathcal{L}_R)}_F^2.
\end{equation}

Unlike in the dense core case, the frame potential of a Khatri-Rao product cannot be separated in terms of the frame potential of its factors. Instead, \eqref{eq:kr_fp} decomposes the frame potential using the Hadamard product of the Grammian of the factors.

\subsubsection{Submodularity of $P(\mathcal{L})$}

Since $P(\mathcal{L})$ does not directly satisfy the conditions [cf. Theorem~\ref{thm:submodular_matroid}]  required for near optimality of the greedy heuristic, we propose using the following set function $Q:2^\mathcal{N}\rightarrow\R$ as a surrogate for the frame potential
\begin{equation}
  Q(\mathcal{S})\coloneqq P(\mathcal{N})-P(\mathcal{N}\setminus\mathcal{S})
  \label{eq:set_func_2}
\end{equation}
with $P(\mathcal{N}) =\norm{\matx{T}_1(\mathcal{N}_1)\circ\dots\circ \matx{T}_r(\mathcal{N}_r)}_F^2$ and 
$P(\mathcal{N}\setminus\mathcal{S}) = \norm{\matx{T}_1(\mathcal{N}_1\setminus\mathcal{S}_1)\circ\dots\circ \matx{T}_R(\mathcal{N}_R\setminus\mathcal{S}_R)}_F^2.$

\begin{theorem}\label{thm:submodular_Q}
   The set function $Q(\mathcal{S})$ defined in \eqref{eq:set_func_2} is a normalized, monotone non-decreasing, submodular function for all subsets of $\mathcal{N}=\bigcup_{i=1}^R\mathcal{N}_i$. 
\end{theorem}
\begin{proof}
  See Appendix~\ref{ap:proof_Q}.
\end{proof}

Using $Q$ and imposing the identifiability constraints defined in Section~\ref{sec:identifiability_constraints} we can write the related optimization problem for the minimization of the frame potential as
\begin{equation}
\label{eq:submodular_max_kr}
\maximize_{\mathcal{S}\subseteq\mathcal{N}}  Q(\mathcal{S}) \quad \text{s. to} \quad \mathcal{S}\in \mathcal{I}_u\cap\mathcal{I}_p
\end{equation}
where $\mathcal{I}_u=\{\mathcal{A}\subseteq\mathcal{N}:|\mathcal{S}|\leq N-L\}$ and $\mathcal{I}_p=\{\mathcal{A}\subseteq\mathcal{N}:|\mathcal{A}\cap\mathcal{N}_i|\leq \beta_i\quad i=1,\dots,R\}$ with $\beta_j = N_j-\max\{K_\text{c},z_j\}$ and $\beta_i=N_i-\max\{1,z_i+1\}\;\; \text{for}\;\;i\neq j$. Here, the choice of $j$ is arbitrary, and can be set depending on the application. For example, with some space-time signals it is more costly to sample space than time, and, in those cases, $j$ is generally chosen for the temporal domain.

This is a submodular maximization problem with a truncated partition matroid constraint [cf. Example~\ref{ex:partition_matroid} from Definition~\ref{def:matroid}]. Thus, from Theorem~\ref{thm:submodular_matroid}, we know that greedy maximization of \eqref{eq:submodular_max_kr} using Algorithm~\ref{alg:greedy_matroid} has a multiplicative near-optimal guarantee.

\begin{corollary}\label{corollary:bound_Q}
The greedy solution $\mathcal{S}_\text{greedy}$ to \eqref{eq:submodular_max_kr} obtained using Algorithm~\ref{alg:greedy_matroid} is $1/2$-near-optimal, i.e., 
$
    Q(\mathcal{S}_\text{greedy})\geq \frac{1}{2}Q(\mathcal{S}^\star). 
 $
 Here, $\mathcal{S}^\star$ denotes the optimal solution of \eqref{eq:submodular_max_kr}.
\end{corollary}

Similar to the dense core case, we can also provide a bound on the near-optimality of of the greedy solution with respect to the frame potential.
\begin{theorem}
The solution set $\mathcal{L}_\text{greedy} = \mathcal{N} \setminus \mathcal{S}_\text{greedy}$ obtained from Algorithm~\ref{alg:greedy_matroid} is near optimal with respect to the frame potential as
$
    P(\mathcal{L}_\text{greedy})\leq \gamma P(\mathcal{L}^\star),
 $
 with $\gamma=0.5\left(\norm{\matx{T}_1(\mathcal{N}_1)\circ\dots\circ \matx{T}_R(\mathcal{N}_R)}_F^2{KL_{\text{\rm min}}^{-2}}+1\right)$ and $\mathcal{L}^\star = \mathcal{N} \setminus~\mathcal{S}^\star$.
\end{theorem}
\begin{proof}
Based on the proof of Theorem~\ref{thm:approx_fp}. The bound is obtained using \eqref{eq:kr_fp} instead of \eqref{eq:kron_fp} in the derivation.
\end{proof}

\subsubsection{Computational complexity}

The computational complexity of the greedy method is now governed by the complexity of computing the Grammian matrices $\matx{T}_i$. This can greatly be improved if before starting the iterations, one precomputes all the outer products in $\{\matx{T}_i\}_{i=1}^R$. Doing this has a computational complexity of $\mathcal{O}(N_\text{max}K_\text{c}^2)$. Then, in every iteration, the evaluation of $P(\mathcal{L})$ would only cost $\mathcal{O}(RK_\text{c}^2)$ operations. Further, because in every iteration we need to query $\mathcal{O}(N_i)$ elements on each domain, and we run the algorithm for $N-L$ iterations, the total time complexity of the iterations is $\mathcal{O}(RN_\text{max}^2K_\text{c}^2)$. This term dominates over the complexity of the precomputations, and thus can be treated as the worst case complexity of the greedy method.

\subsubsection{Practical considerations}

The proposed scheme suffers from the same issues as in the dense core case. Namely, it tends to empty the domains sequentially, thus producing solutions which are tight on the identifiability constraints. Nevertheless, as we indicated for the dense core, the drop in performance associated with the proximity of the solutions to the constraints can be reduced by giving some slack to the constraints.

%% file: results.tex
 \begin{figure*}
   \centering
    \includegraphics[width=\textwidth]{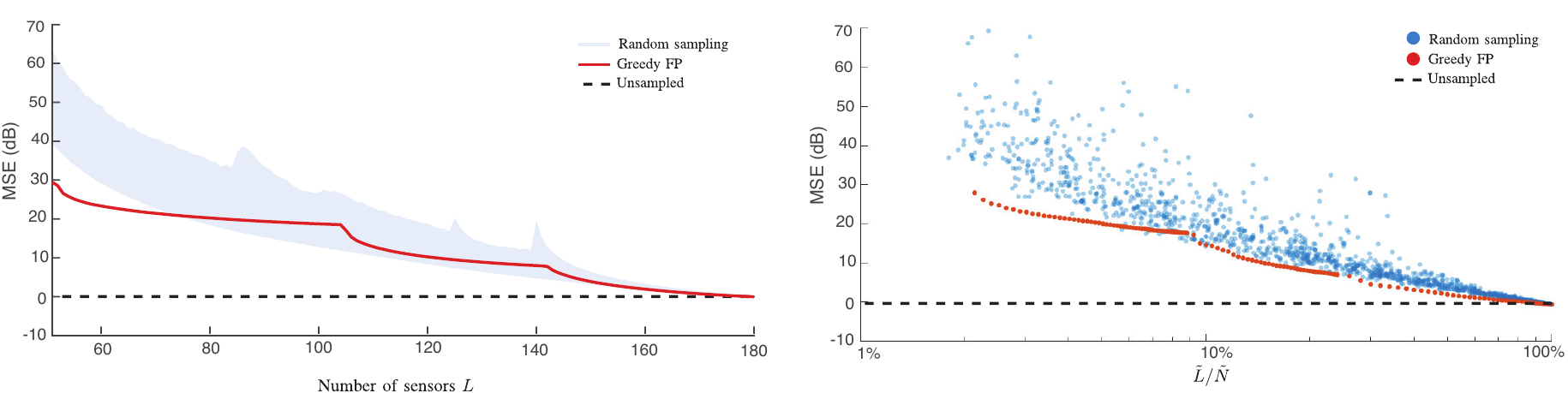}
   \caption{\footnotesize{Dense core with $R=3$, $N_1=50,N_2=60,N_3=70$, $K_1=10,K_2=20,K_3=15$, and $\alpha_1=\alpha_2=\alpha_3=2$.}}
   \label{fig:kron_comp}
      \vspace{-4mm}
 \end{figure*}
  \begin{figure*}
   \centering
        \includegraphics[width=\textwidth]{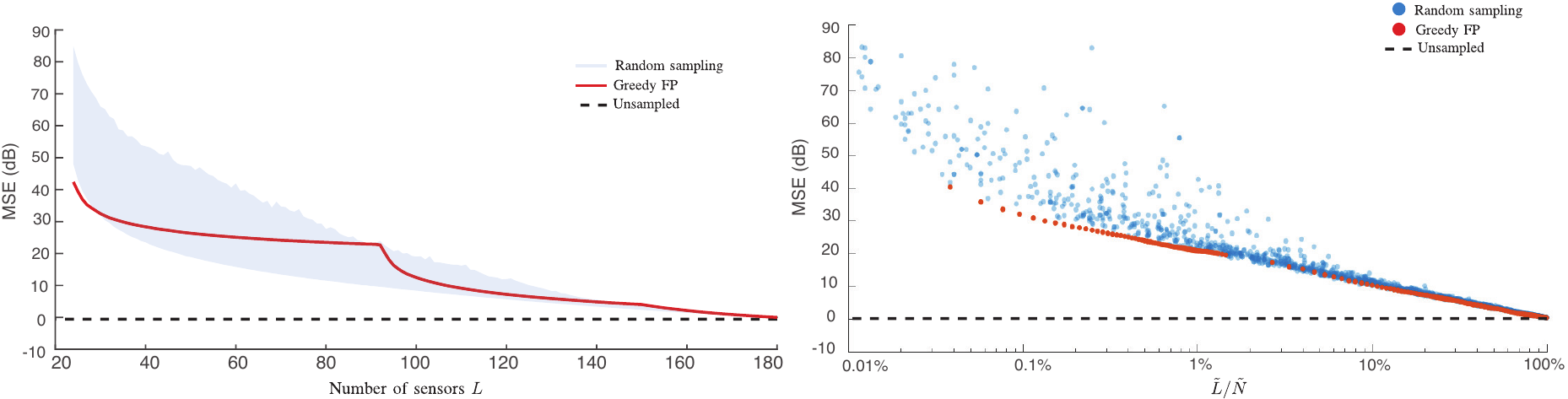}
   \caption{\footnotesize{Diagonal core with $R=3$ with $N_1=50,N_2=60,N_3=70$, $K_\text{c}=20$, $\beta_1=\beta_2=1$, and $\beta_3=20$.}}
   \label{fig:diag_comp}
   \vspace{-3mm}
 \end{figure*}

\section{Numerical results}\label{sec:numerical_results}

In this section\footnote{The code to reproduce these experiments can be found at https://gitlab.com/gortizji/sparse\_tensor\_sensing.}, we will illustrate the developed framework through several examples. First, we will show some results obtained on synthetic datasets to compare the performance of the different near-optimal algorithms. Then, we will focus on large-scale real-world examples related to (i) graph signal processing: \emph{sampling product graphs for active learning in recommender systems},  and (ii) array processing for wireless communications: \emph{multiuser source separation}, to show the benefits of the developed framework.

\subsection{Synthetic example}
\subsubsection{Dense core}

We compare the performance in terms of the theoretical MSE of our proposed greedy algorithm (henceforth referred to as greedy-FP) to a random sampling scheme based on randomly selecting rows of $\matx{U}_i$ such that the resulting subset of samples also have a Kronecker structure. Only those selections that satisfy the identifiability constraints in \eqref{eq:identifiability_constraint} are considered  valid. We note that the time complexity of evaluating $M$ times the MSE for a Kronecker-structured sampler is $\mathcal{O}(MN_\text{max}^2K_\text{max})$. For this reason, using many realizations (say, a large number $M$)  of random sampling to obtain a good sparse sampler is computationally intense.

To perform this comparison, we draw $M=100$ realizations of three random Gaussian matrices $\{\matx{U}_i\in\R^{N_i\times K_i}\}_{i=1}^{R=3}$ with dimensions $N_1=50$, $N_2=60$, and $N_3=70$. For each of these models, we solve $\eqref{eq:submodular_max_kron_matroid}$ for different number of sensors $L$ using greedy-FP. We also compute $M=100$ realizations of random sampling for each $L$. Fig.~\ref{fig:kron_comp} shows the results of these experiments. The plot on the left shows the performance averaged over the different models against the number of sensors, wherein the blue shaded area represents the 10-90 percentile average interval of the random sampling scheme. The performance values, in dB scale, are normalized by the value of the unsampled MSE. Because the estimation performance is heavily influenced by its related number of samples $\tilde{L}$, and noting the fact that a value of $L$ may lead to different $\tilde{L}$, we also present, in the plot on the right side of Fig.~\ref{fig:kron_comp}, the performance comparison for one model realization against the relative number of samples $\tilde{L}/\tilde{N}$ so that differences in the informative quality of the selections are highlighted.


The plots in Fig.~\ref{fig:kron_comp} illustrate some important features of the proposed sparse sampling method. When comparing the performance against the number of sensors, we see that there are areas where greedy-FP performs as well as random sampling. However, when comparing the same results against the number of samples we see that greedy-FP consistently performs better than random sampling. The reason for this disparity is due to characteristics of greedy-FP that we introduced in Section~\ref{sec:dense_fp}. Namely, the tendency of greedy-FP to produce sampling sets with the minimum number of samples.

On the other hand, the performance curve of greedy-FP shows three bumps (recall that we use $R=3$). Again, this is a consequence of greedy-FP trying to meet the identifiability constraints in \eqref{eq:submodular_max_kron_matroid} with equality. As we increase $L$, the solutions of greedy-FP increase in cardinality by adding more elements to a single domain until the constraints are met, and then proceed to the next domain. The bumps in Fig.~\ref{fig:kron_comp} correspond precisely to these instances.


\subsubsection{Diagonal core}

We perform the same experiment for the diagonal core case. The results are shown in Fig.~\ref{fig:diag_comp}. Again we see that the proposed algorithm outperforms random sampling, especially when collecting just a few samples. Furthermore, as happened in the dense core case, the performance curve of greedy-FP follows a stairway shape.

\subsection{Active learning for recommender systems}

Current recommendation algorithms seek solving an estimation problem of the form: given the past recorded preferences of a set of users, what is the \emph{rating} that these would give to a set of products? In this paper, in contrast, we focus on the data acquisition phase of the recommender system, which is also referred to as active learning/sampling. In particular, we claim that by carefully designing which users to poll and on which items, we can obtain an estimation performance on par with the state-of-the-art methods, but using only a fraction of the data that current methods require, and using a simple least-squares estimator.

We showcase this idea on the MovieLens $100k$ dataset \cite{movielens} that contains partial ratings of $N_1=943$ users over $N_2=1682$ movies which are stored in a second-order tensor $\mathbfcal{X}\in\R^{N_1\times N_2}$. At this point, we emphasize the need for our proposed framework, since it is obvious that designing an unstructured sampling set with about 1.5 million candidate locations is unfeasible with current computing resources.


A model of $\mathbfcal{X}$ in the form of \eqref{eq:tensor_linear_model} can be obtained by viewing $\mathbfcal{X}$ as a signal that lives on a graph. In particular, the first two modes of $\mathbfcal{X}$ can be viewed as a signal defined on the Cartesian product of a user and movie graph, respectively. These two graphs, shown in Fig.~\ref{fig:graphs}, are provided in the dataset and are two 10-nearest-neighbors graphs created based on the user and movie features. 

Based on the recent advances in graph signal processing (GSP)~\cite{moura,gsp}, $\mathbfcal{X}$ can be decomposed as $\mathbfcal{X}=\mathbfcal{X}_\text{f}\bullet_1 \matx{V}_1\bullet_2\matx{V}_2$. Here, $\matx{V}_1\in\mathbb{R}^{N_1\times N_1}$ and $\matx{V}_2\in\mathbb{R}^{N_2\times N_2}$ are the eigenbases of the Laplacians of the user and movie graphs, respectively, and $\mathbfcal{X}_\text{f}\in\matx{C}^{N_1\times N_2}$ is the so-called graph spectrum of $\mathbfcal{X}$~\cite{moura,gsp}. Suppose the energy of the spectrum of $\mathbfcal{X}$ is concentrated in the first few $K_1$ and $K_2$ columns of $\matx{V}_1$ and $\matx{V}_2$, respectively, then $\mathbfcal{X}$ admits a low-dimensional representation, or $\mathbfcal{X}$ is said to be smooth or bandlimited with respect to the underlying graph~\cite{gsp}. This property has been exploited in~\cite{kalofolias2014matrix,marques} to impute the missing entries in $\mathbfcal{X}$. In contrast, we propose a scheme for sampling and reconstruction of signals defined on product graphs.

 In our experiments, we set $K_1=K_2=20$, and obtain the decomposition $\mathbfcal{X}=\mathbfcal{G}\bullet_1 \matx{U}_1\bullet_2\matx{U}_2$, where $\matx{U}_1\in\C^{N_1\times K_1}$ and $\matx{U}_2\in\C^{N_2\times K_1}$ consist of the first $K_1$ and $K_2$ columns of $\matx{V}_1$ and $\matx{V}_2$, respectively; and $\mathbfcal{G}=\mathbfcal{X}(1:K_1,1:K_2)$.

\begin{figure}
    \centering
    \begin{subfigure}[b]{0.49\columnwidth}
        \includegraphics[width=\textwidth]{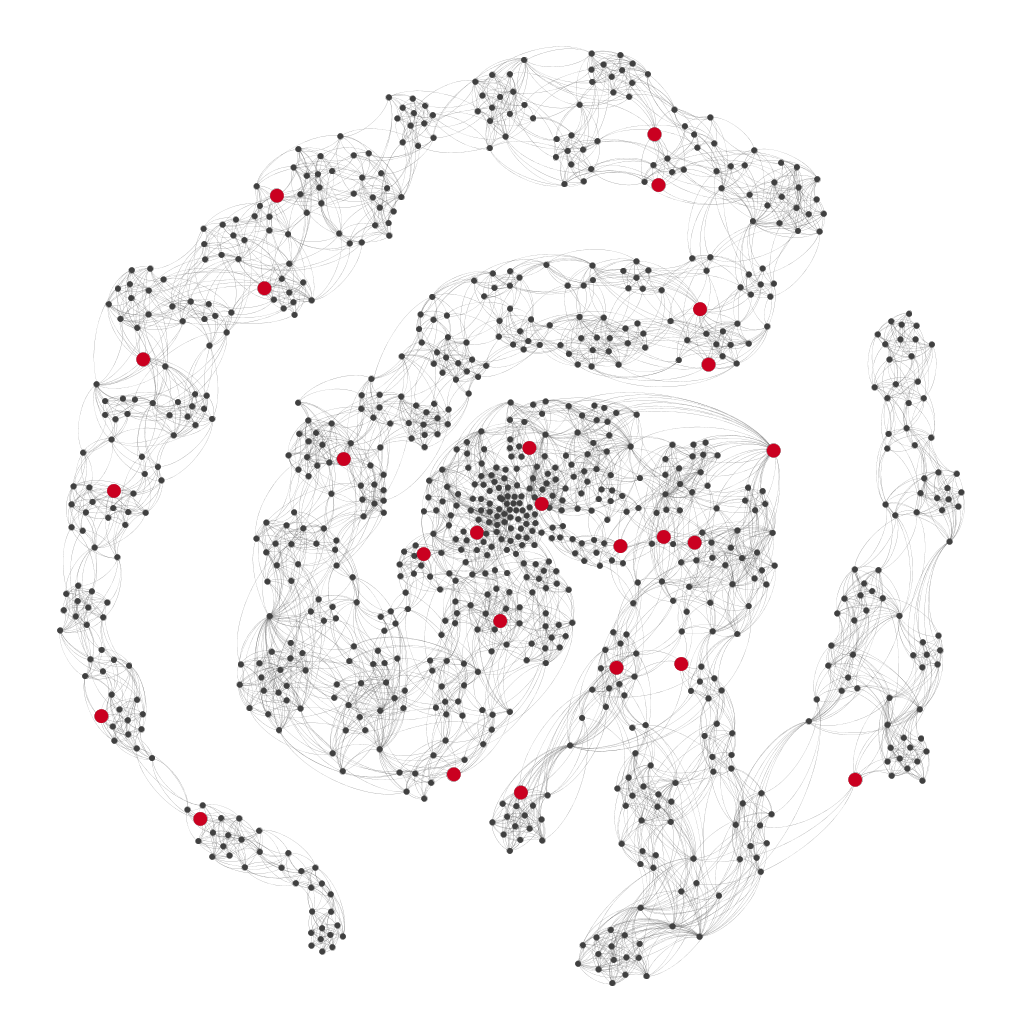}
        \caption{User graph}
        \label{fig:users}
    \end{subfigure}
    \hfill
    \begin{subfigure}[b]{0.49\columnwidth}
        \includegraphics[width=\textwidth]{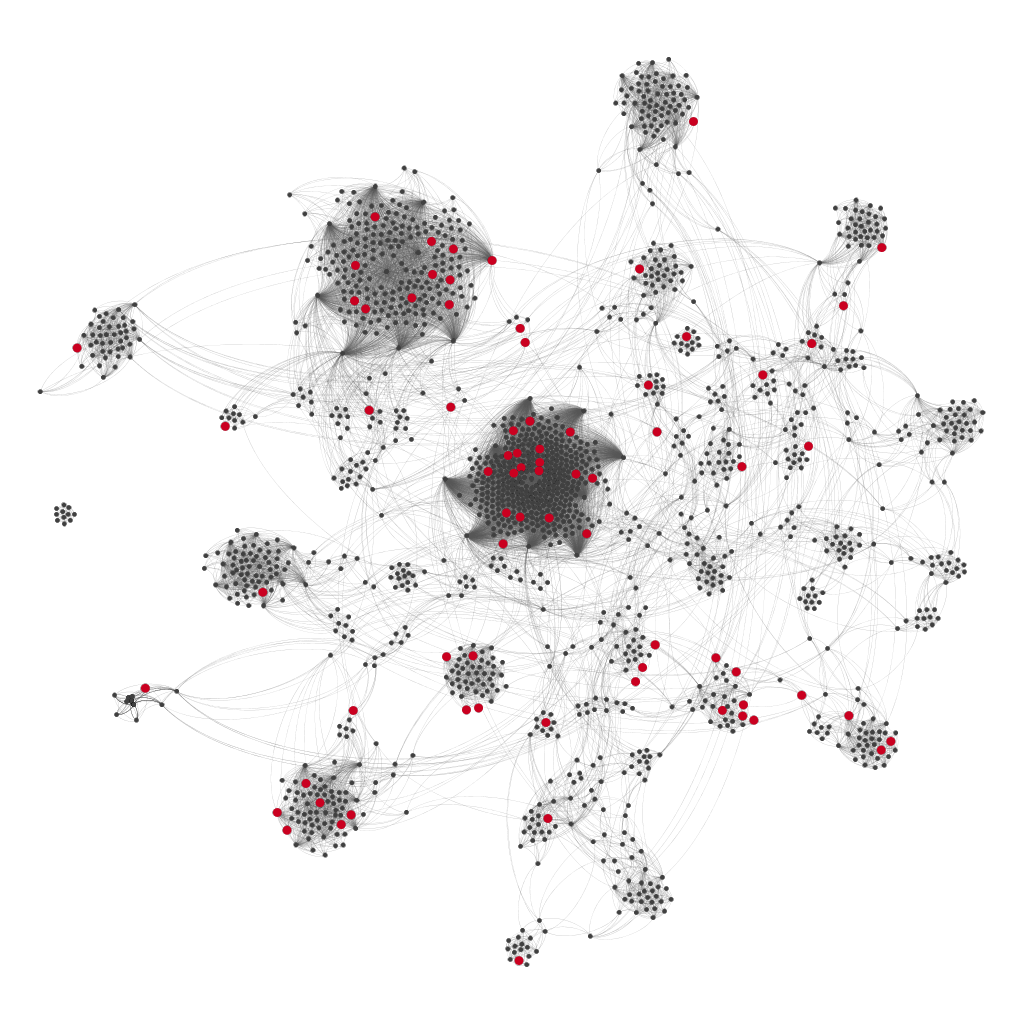}
        \caption{Movie graph}
        \label{fig:movies}
    \end{subfigure}
    \caption{\footnotesize{User and movie networks. The red (black) dots represent the observed (unobserved) vertices. Visualization obtained using Gephi \cite{gephi}.}}
    \label{fig:graphs}
    \vspace{-2mm}
\end{figure}

For the greedy algorithm we use $L = 100$ and $\alpha_1=\alpha_2=5$, resulting in a selection of $L_1 = 25$ users and $L_2 = 75$ movies, i.e., a total of $1875$ vertices in the product graph. Fig.~\ref{fig:graphs}, shows the sampled users and movies, i.e., users to be probed for movie ratings. The user graph [cf. Fig.~\ref{fig:users}] is made out of small clusters connected in a chain-like structure, resulting in a uniformly spread distribution of observed vertices. On the other hand, the movies graph [cf. Fig.~\ref{fig:movies}] is made out of a few big and small clusters. Hence, the proposed active querying scheme assigns more observations to the bigger clusters and fewer observations to the smaller ones.

To evaluate the performance of our algorithm, we compute the RMSE of the estimated data using the test mask provided by the dataset. Nevertheless, since our active query method requires access to ground truth data (i.e., we need access to the samples at locations suggested by the greedy algorithm) which is not provided in the dataset, we use GRALS \cite{grals} to complete the matrix, and use its estimates when required. A comparison of our algorithm to the performance of the state-of-the-art methods run on the same dataset is shown in Table~\ref{tab:performance}. In light of these results, it is clear that a proper design of the sampling set allows to obtain top performance with significantly fewer ratings, i.e., about an order of magnitude, and using a much simpler non-iterative estimator.

\begin{table}
\centering{}
\begin{tabular}{@{}ccc@{}}
\toprule
\textbf{Method} & \textbf{Number of samples} & \textbf{RMSE} \\ \midrule
GMC \cite{kalofolias2014matrix}            & 80,000                      & 0.996         \\ \midrule
GRALS  \cite{grals}         & 80,000                      & 0.945         \\ \midrule
sRGCNN\cite{monti2017geometric}         & 80,000                      & 0.929         \\ \midrule
GC-MC    \cite{berg2017graph}       & 80,000                      & \textbf{0.905}         \\ \midrule\midrule
Our method      & \textbf{1,875}                       & 0.9347        \\ \bottomrule
\end{tabular}
\caption{\footnotesize{Performance on MovieLens $100k$. Baseline scores are taken from \cite{berg2017graph}.}}
\label{tab:performance}
\end{table}

\subsection{Multiuser source separation}

In multiple-input multiple-output (MIMO) communications \cite{mimo_magazine}, the use of rectangular arrays \cite{ura} allows to separate signals coming from different azimuth and elevation angles, and it is common that users transmit data using different spreading codes to reduce the interference from other sources. Reducing hardware complexity by minimizing the number of antennas and samples to be processed is an important concern in the design of MIMO receivers. This design can be seen as a particular instance of sparse tensor sampling. 

We consider a scenario with $K_\text{c}$ users located at different angles of azimuth ($\phi$) and elevation ($\theta$) transmitting using unique spreading sequences of length $N_3$. The receiver consists of a uniform rectangular array (URA) with antennas located on a $N_1\times N_2$ grid. Each time instant, every antenna receives \cite{ura}
\begin{align*}
  x(r,l,m,n) &= \sum_{k=1}^{K_\text{c}} s_k(r) c_k(l) e^{j2\pi n \Delta_x \sin \theta_k}e^{j 2\pi m \Delta_y \sin \phi_k} \\
  &\hspace{1em}+ w(r,l,m,n),
\end{align*}
where $s_k(r)$ the symbol transmitted by user $k$ in the $r$th symbol period; $c_k(l)$ the $l$th sample of the spreading sequence of the $k$th user; $\Delta_x$ and $\Delta_y$ the antenna separations in wavelengths of the URA in the $x$ and $y$ dimensions, respectively; and $\phi_{k}$ and $\theta_k$ the azimuth and elevation coordinates of user $k$, respectively; and where $w(r,l,m,n)$ represents an additive white Gaussian noise term with zero mean and variance $\sigma^2$. For the $r$-th symbol period, all these signals can be collected in a 3rd-order tensor $\mathbfcal{X}(r)\in\C^{N_1\times N_2\times N_3}$ that can be decomposed as
\begin{equation*}
  \mathbfcal{X}(r)=\mathbfcal{S}(r)\bullet_1 \matx{U}_1\bullet_2 \matx{U}_2 \bullet_3 \matx{U}_3+\mathbfcal{W}(r)
\end{equation*}
where $\matx{U}_1\in\C^{N_1\times K_\text{c}}$ and $\matx{U}_2\in\C^{N_2\times K_\text{c}}$ are the array responses for the $x$ and $y$ directions, respectively; $\matx{U}_3\in\C^{N_3\times K_\text{c}}$ contains the spreading sequences of all users in its columns; and $\mathbfcal{S}(r)\in\C^{K_\text{c}\times K_\text{c}\times K_\text{c}}$ is a diagonal tensor that stores the symbols of all users for the $r$th symbol period on its diagonal. 

\begin{figure}[t]
  \centering
  \includegraphics[width=\columnwidth]{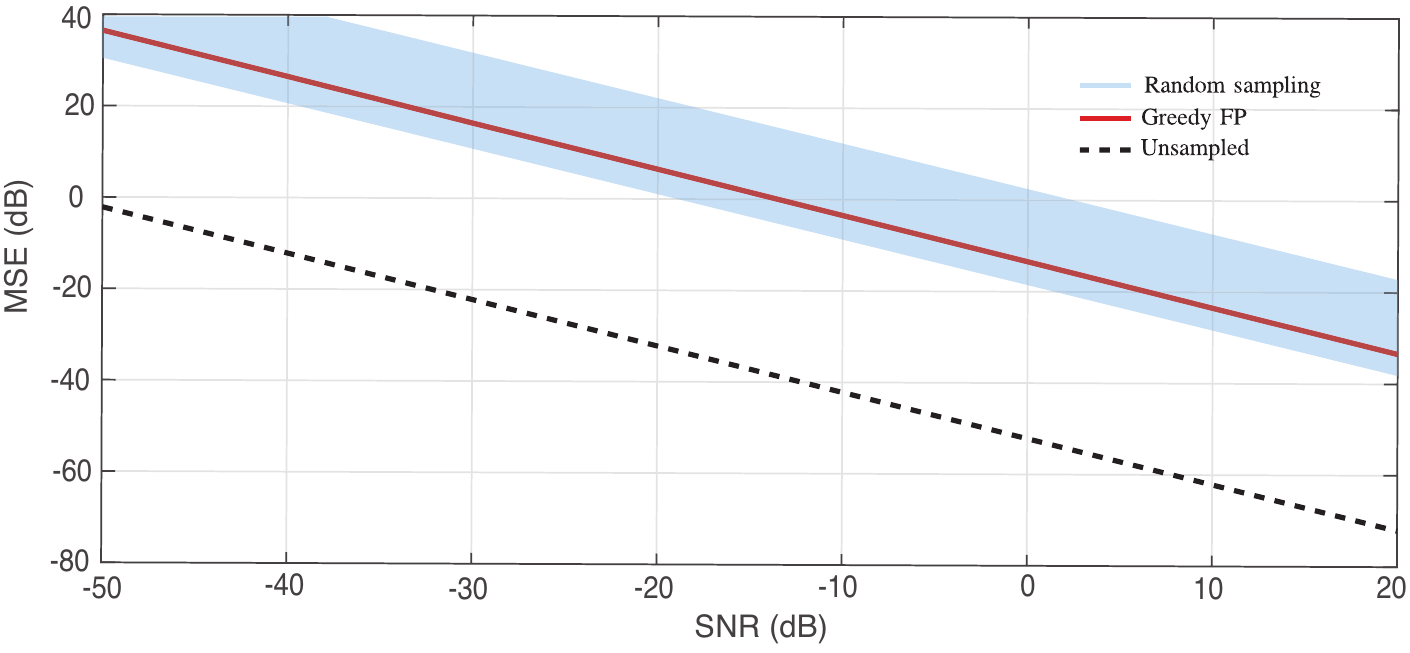}
  \caption{\footnotesize{MSE of symbol reconstruction. $N_1=50$, $N_2=60$, $N_3=100$, and $L=15$.}}
  \label{fig:error_prob}
  \vspace{-5mm}
\end{figure}

We simulate this setup using $K_\text{c}=10$ users that transmit BPSK symbols with different random powers and that are equispaced in azimuth and elevation. We use a rectangular array with $N_1=50$ and $N_2=60$ for the ground set locations of the antennas, and binary random spreading sequences of length $N_3=100$. With these parameters, each $\mathbfcal{X}(r)$ has $300,000$ entries. We generate many realizations of these signals for different levels of signal-to-noise ratio (SNR) and sample the resulting tensors using the greedy algorithm for the diagonal core case with $L=15$, resulting in a relative number of samples of $0.048\%$. The results are depicted in Fig.~\ref{fig:error_prob}, where the blue shaded area represents the MSE obtained with the best and worst random samplers. As expected, the MSE of the reconstruction decreases exponentially with the SNR. For a given MSE, achieving maximum compression requires transmitting with a higher SNR of about 30dB than the one needed for no compression. Besides, we see that our proposed greedy algorithm consistently performs as well as the best random sampling scheme.

%% file: conclusions.tex
\section{Conclusions}\label{sec:conclusions}

In this paper, we presented the design of sparse samplers for inverse problems with tensors. We have seen that by using samplers with a Kronecker structure we can overcome the curse of dimensionality, and design efficient subsampling schemes that guarantee a good performance for the reconstruction of multidomain tensor signals. We presented sparse sampling design methods for cases in which the multidomain signals can be decomposed using a multilinear model with a dense core or a diagonal core. For both cases, we have provided a near-optimal greedy algorithm based on submodular optimization methods to compute the sampling sets.  

%% file: appendix.tex
\appendix

\subsection{Proof of Theorem~\ref{thm:submodular_G}}\label{ap:proof_G}

In order to simplify the derivations, let us introduce the notation
$
  \bar{F}_i(\mathcal{S}_i)=F_i(\mathcal{N}_i\setminus\mathcal{S}_i),
$
so that $G(\mathcal{S})$ can also be written
\begin{equation}
  G(\mathcal{S})\coloneqq\prod_{i=1}^R F_i(\mathcal{N}_i)-\prod_{i=1}^R\bar{F}_i(\mathcal{S}_i) \label{eq:def_fp_kron}.
\end{equation}

From \eqref{eq:def_fp_kron} it is evident that $G(\varnothing)=0$. Thus, proving that $G$ is normalized. To prove monotonicity, recall that the single domain frame potential terms $F_i(\mathcal{L}_i)$ are all non-negative, monotone non-decreasing functions for all $\mathcal{L}_i\subseteq\mathcal{N}_i$ \cite{frame_potential}. Therefore, $\bar{F}_i(\mathcal{S}_i)=F_i(\mathcal{N}\setminus\mathcal{S}_i)$ will be non-negative, but monotone non-increasing. Let $\mathcal{S}\subseteq\mathcal{N}$ and $x\in\mathcal{N}\setminus\mathcal{S}$. Without loss of generality, let us assume $x\in\mathcal{N}_i$. Then, we have
\begin{align*}
  &G(\mathcal{S}\cup\{x\})=\prod_{i=1}^RF_i(\mathcal{N}_i)-\bar{F}_i(\mathcal{S}_i\cup\{x\})\prod_{j\neq i} \bar{F}_j(\mathcal{S}_j),\\
  &G(\mathcal{S})=\prod_{i=1}^RF_i(\mathcal{N}_i)-\bar{F}_i(\mathcal{S}_i)\prod_{j\neq i} \bar{F}_j(\mathcal{S}_j).
\end{align*}
Now, since $\bar{F}_i(\mathcal{S}_i)\geq\bar{F}_i(\mathcal{S}_i\cup\{x\})$, we know that
$
  G(\mathcal{S}\cup\{x\})\geq G(\mathcal{S}).
$
Hence, $G(\mathcal{S})$ is monotone non-decreasing.

To prove submodularity, recall that every $F_i(\mathcal{L}_i)$ is supermodular~\cite{frame_potential}. As taking the complement preserves (super)submodularity, $\bar{F}_i(\mathcal{L}_i)=F_i(\mathcal{N}_i\setminus\mathcal{L}_i)$ is also supermodular. Let $\mathcal{S}=\bigcup_{i=1}^R\mathcal{A}_i$, with $\mathcal{A}_i\subseteq\mathcal{N}_i$ for $i=1,\dots,R$, such that $\{\mathcal{A}_i\}_{i=1}^R$, forms a partition of $\mathcal{S}$. Now, recall from Definition~\ref{def:submodular} that for $G$ to be submodular we require that $\forall x,y\in\mathcal{N}\setminus \mathcal{S}$
\begin{equation}
  G(\mathcal{S}\cup\{x\})-G(\mathcal{S})\geq G(\mathcal{S}\cup\{x,y\})-G(\mathcal{S}\cup\{y\}).\label{eq:submodular_def_rep}
\end{equation}
As the ground set is now partitioned into the union of several ground sets, there are two possible ways the elements $x$ and $y$ can be selected. Either they both belong to the same domain, or they belong to different domains. We next prove that \eqref{eq:submodular_def_rep} is satisfied for the aforementioned both cases.

Suppose $x,y\in\mathcal{N}_i$, then \eqref{eq:submodular_def_rep} can be developed as
  \begin{align*}
    \bar{F}_i&(\mathcal{A}_i)\prod_{j\neq i}\bar{F}_j(\mathcal{A}_j)-\bar{F}_i(\mathcal{A}_i\cup\{x\})\prod_{j\neq i}\bar{F}_j(\mathcal{A}_j)\\
    &\geq \bar{F}_i(\mathcal{A}_i\cup\{y\})\prod_{j\neq i}\bar{F}_j(\mathcal{A}_j)- \bar{F}_i(\mathcal{A}_i\cup\{i,j\})\prod_{j\neq i}\bar{F}_j(\mathcal{A}_j),
  \end{align*}
which can be further simplified to
\begin{align*}
    \bar{F}_i(\mathcal{A}_i\cup\{x\})-\bar{F}_i(\mathcal{A}_i)
    \leq \bar{F}_i(\mathcal{A}_i\cup\{x,y\})-\bar{F}_i(\mathcal{A}_i\cup\{y\}).
\end{align*}
The above inequality is true since $\bar{F}_i$ is supermodular.

Next, suppose $x\in\mathcal{N}_i$ and $y\in\mathcal{N}_j$ with $i\neq j$, then \eqref{eq:submodular_def_rep} can be expanded as
  \begin{align*}
    \prod_{k\neq i, j}&\bar{F}_k(\mathcal{A}_k)\left[\bar{F}_i(\mathcal{A}_i)\bar{F}_j(\mathcal{A}_j)-\bar{F}_i(\mathcal{A}_i\cup\{x\})\bar{F}_j(\mathcal{A}_j)\right]\\
    &\geq \prod_{k\neq i, j}\bar{F}_k(\mathcal{A}_k)\left[\bar{F}_i(\mathcal{A}_i)\bar{F}_j(\mathcal{A}_j\cup\{y\})\right.
    \\&\quad\left.- \bar{F}_i(\mathcal{A}_i\cup\{x\})\bar{F}_j(\mathcal{A}_j\cup\{y\})\right].
  \end{align*}
  Extracting the common factors
  \begin{equation}
    \left[\bar{F}_i(\mathcal{A}_i)-\bar{F}_i(\mathcal{A}_i\cup\{x\})\right]\left[\bar{F}_j(\mathcal{A}_j)-\bar{F}_j(\mathcal{A}_j\cup\{y\})\right]\geq 0.\label{eq:cond_kron_3}
  \end{equation}
  Since $\bar{F}_i$ and $\bar{F}_j$ are non-increasing 
  \begin{gather*}
    \bar{F}_i(\mathcal{A}_i)-\bar{F}_i(\mathcal{A}_i\cup\{x\})\geq 0; \quad 
    \bar{F}_j(\mathcal{A}_j)-\bar{F}_j(\mathcal{A}\cup\{y\})\geq 0.
  \end{gather*}
  Thus, \eqref{eq:cond_kron_3} is always satisfied, thus proving that \eqref{eq:submodular_def_rep} is satisfied for any $\mathcal{S}\subseteq\mathcal{N}$ and $x,y\in\mathcal{N}\setminus\mathcal{S}$ and therefore $G$ is submodular.

\subsection{Proof of Theorem~\ref{thm:submodular_Q}}\label{ap:proof_Q}

We  divide the proof in two parts. First, we derive some properties of the involved operations that are useful to simplify the proof. Then, we  use this to derive the proof.

\subsubsection{Preliminaries}
First, note that the single-domain Grammian matrices satisfy the following lemma.
\begin{lemma}[Grammian of disjoint union]\label{lemma:union_T}
  Let $\mathcal{X},\mathcal{Y}\subseteq\mathcal{N}_i$ with $\mathcal{X}\cap\mathcal{Y}=\varnothing$. Then, the Grammian of $\mathcal{X}\cup\mathcal{Y}$ satisfies
  \begin{equation*}
    \matx{T}_i(\mathcal{X}\cup\mathcal{Y})=\matx{T}_i(\mathcal{X})+\matx{T}_i(\mathcal{Y}).
  \end{equation*}
\end{lemma}
\begin{proof}
  Let $\vec{u}_{i,j}$ denote the $j$th row of $\matx{T}_i$. Then,
  \begin{equation*}
    \matx{T}_i(\mathcal{X}\cup\mathcal{Y})=\sum_{j\in\mathcal{X}\cup\mathcal{Y}}\norm{\vec{u}_{i,j}}^2_2=\sum_{j\in\mathcal{X}}\norm{\vec{u}_{i,j}}^2_2+\sum_{j\in\mathcal{Y}}\norm{\vec{u}_{i,j}}^2_2.
  \end{equation*}
\end{proof}

Let us introduce the complement Grammian matrix
\begin{equation}
  \bar{\matx{T}}_i(\mathcal{S}_i)\coloneqq\matx{T}_i(\mathcal{N}_i\setminus\mathcal{S}_i)=\matx{T}_i(\mathcal{N}_i)-\matx{T}_i(\mathcal{S}_i),\label{eq:tilde_t}
\end{equation}
which satisfies the following lemma.
\begin{lemma}[Complement Grammian of disjoint union]\label{lemma:union_T_bar}
  Let $\mathcal{X},\mathcal{Y}\subseteq\mathcal{N}_i$ with $\mathcal{X}\cap\mathcal{Y}=\varnothing$. Then,
$
    \bar{\matx{T}}_i(\mathcal{X}\cup\mathcal{Y})=\bar{\matx{T}}_i(\mathcal{X})-\matx{T}_i(\mathcal{Y}).
 $
\end{lemma}
\begin{proof}
    From \eqref{eq:tilde_t} and Lemma~\ref{lemma:union_T}, we have
    \begin{align*}
      \bar{\matx{T}}_i(\mathcal{X}\cup\mathcal{Y})&=\matx{T}_i(\mathcal{N}_i)-\left[\matx{T}_i(\mathcal{X})+\matx{T}_i(\mathcal{Y})\right]=\bar{\matx{T}}_i(\mathcal{X})-\matx{T}_i(\mathcal{Y}).
    \end{align*}
\end{proof}

Now, let us introduce an operator to compress the writing of the multidomain Hadamard product
\[
  \mathbb{T}(\mathcal{L})\coloneqq\matx{T}_1(\mathcal{L}_1)\circ\dots\circ\matx{T}_R(\mathcal{L}_R),
\]
or alternatively for the complement Grammian
\[
  \bar{\mathbb{T}}(\mathcal{S})\coloneqq\bar{\matx{T}}_1(\mathcal{S}_1)\circ\dots\circ\bar{\matx{T}}_R(\mathcal{S}_R).
\]
Furthermore, we  write the Hadamard multiplication of all $\matx{T}_i$ with $i=1,\dots,R$, but $j$ as
\begin{equation*}
  \mathbb{T}_{-j}(\mathcal{L})\coloneqq\mathbb{T}(\mathcal{L})\circ\matx{T}_j(\mathcal{L}_j)^{\circ -1},
\end{equation*}
where $\matx{A}^{\circ n}$ denotes the element-wise $n$th power of $\matx{A}$. Similarly, for the complement Grammians, we will use $\bar{\mathbb{T}}_{-i}(\mathcal{S})$. We also make use of the following properties of the Hadamard product.

\begin{property}
  The Hadamard product of two positive semidefinite matrices is always positive semidefinite.
\end{property}
\begin{property}\label{prop:frob}
  Let $\matx{A},\matx{B}\in\C^{N\times N}$. Then, 
  \begin{equation*}
      \norm{\matx{A}\circ\matx{B}}^2_F=\tr{\matx{A}^{\circ2}\left(\matx{B}^{\circ2}\right)^T}=\braket{\matx{A}^{\circ2},\matx{B}^{\circ2}}.
  \end{equation*}
\end{property}
Let us introduce the notation
\begin{equation}
  \matx{H}_i(\mathcal{S})\coloneqq\matx{T}_i^{\circ 2}(\mathcal{S})\qquad\text{and}\qquad\bar{\matx{H}}_i(\mathcal{S})\coloneqq\bar{\matx{T}}_i^{\circ 2}(\mathcal{S}),
\end{equation}
which satisfies the following lemma.
\begin{lemma}\label{lemma:union_H}
  Let $\mathcal{X},\mathcal{Y}\subseteq\mathcal{N}_i$ with $\mathcal{X}\cap\mathcal{Y}=\varnothing$. Then,
  \begin{align*}
    \matx{H}_i(\mathcal{X}\cup\mathcal{Y})&=\matx{T}_i^{\circ 2}(\mathcal{X}\cup\mathcal{Y})=\left(\matx{T}_i(\mathcal{X})+\matx{T}_i(\mathcal{Y})\right)^{\circ 2}\\
    &=\matx{H}_i(\mathcal{X})+\matx{H}_i(\mathcal{Y})+2\matx{T}_i(\mathcal{X})\circ \matx{T}_i(\mathcal{Y}).
  \end{align*}
and
  \begin{align*}
    \bar{\matx{H}}_i(\mathcal{X}\cup\mathcal{Y})&=\bar{\matx{T}}_i^{\circ 2}(\mathcal{X}\cup\mathcal{Y})=\left(\bar{\matx{T}}_i(\mathcal{X})-\matx{T}_i(\mathcal{Y})\right)^{\circ 2}\\
    &=\bar{\matx{H}}_i(\mathcal{X})+\matx{H}_i(\mathcal{Y})-2\bar{\matx{T}}_i(\mathcal{X})\circ \matx{T}_i(\mathcal{Y}).
  \end{align*}
\end{lemma}

Moreover, as we did with the Grammian matrices, we introduce the notation
\begin{equation*}
  \mathbb{H}(\mathcal{L})\coloneqq \matx{H}_1(\mathcal{L}_1)\circ\dots\circ\matx{H}_R(\mathcal{L}_R),
\end{equation*}
and
\begin{equation*}
  \mathbb{H}_{-j}(\mathcal{L})\coloneqq \mathbb{H}(\mathcal{L})\circ\matx{H}_j(\mathcal{L}_j)^{\circ-1},
\end{equation*}
with its analogue $\bar{\mathbb{H}}$, and $\bar{\mathbb{H}}_{-j}$. Due to Property 1, all these matrices are also positive semidefinite.

Finally, note that with the new notation we can simplify the definition of $Q$ to
\begin{equation}
  Q(\mathcal{S})\coloneqq\norm{\mathbb{T}(\mathcal{N})}_F^2-\norm{\bar{\mathbb{T}}(\mathcal{S})}_F^2.
\end{equation}

\subsubsection{Derivation}
Normalization is derived from the fact that $\bar{\matx{T}}_i(\varnothing)=\matx{T}_i(\mathcal{N})$. To prove monotonicity, let $\mathcal{S}\subseteq\mathcal{N}$ and $x\in\mathcal{N}\setminus\mathcal{S}$. Without loss of generality, assume $x\in\mathcal{N}_i$. We have
\begin{align*}
  &Q(\mathcal{S}\cup\{x\})=\norm{\mathbb{T}(\mathcal{N})}_F^2-\norm{\bar{\matx{T}}_i(\mathcal{S}_i\cup\{x\})\circ\bar{\mathbb{T}}_{-i}(\mathcal{S})}_F^2,\\
  &Q(\mathcal{S})=\norm{\mathbb{T}(\mathcal{N})}_F^2-\norm{\bar{\matx{T}}_i(\mathcal{S}_i)\circ\bar{\mathbb{T}}_{-i}(\mathcal{S})}_F^2.
\end{align*}
Monotonicity requires that $Q(\mathcal{S}) \leq Q(\mathcal{S}\cup\{x\})$, or
\begin{align*}
 -\norm{\bar{\matx{T}}_i(\mathcal{S}_i)\circ\bar{\mathbb{T}}_{-i}(\mathcal{S})}_F^2 &\leq-\norm{\bar{\matx{T}}_i(\mathcal{S}_i\cup\{x\})\circ\bar{\mathbb{T}}_{-i}(\mathcal{S})}_F^2.
\end{align*}
Using Property~\ref{prop:frob}, we have
\begin{equation*}
  \braket{\bar{\matx{T}}_i(\mathcal{S}_i),\bar{\mathbb{T}}_{-i}(\mathcal{S})}\geq\braket{\bar{\matx{T}}_i(\mathcal{S}_i\cup\{x\}),\bar{\mathbb{T}}_{-i}(\mathcal{S})}.
\end{equation*}
Expanding the unions using Lemma~\ref{lemma:union_T_bar}, and due to the linearity of the inner product this becomes
\begin{equation*}
  0\leq\braket{\matx{T}_i(\mathcal{S}_i\cup\{x\}),\bar{\mathbb{T}}_{-i}(\mathcal{S})},
\end{equation*}
which is always satisfied because the inner product between two positive semidefinite matrices is always greater or equal than zero.

To prove submodularity, let  $\mathcal{S}=\bigcup_{i=1}^R\mathcal{A}_i$, with $\mathcal{A}_i\subseteq\mathcal{N}_i$ for $i=1,\dots,R$ such that $\{\mathcal{A}_i\}_{i=1}^R$, forms a partition of $\mathcal{S}$. For $Q$ to be submodular we require that $\forall x,y\in\mathcal{N}\setminus\mathcal{S}$
\begin{equation}
  Q(\mathcal{S}\cup\{x\})-Q(\mathcal{S})\geq Q(\mathcal{S}\cup\{x,y\})-Q(\mathcal{S}\cup\{y\}).\label{eq:submodular_Q}
\end{equation}

As before, we have two different cases. Suppose  $x,y\in\mathcal{N}_i$, then \eqref{eq:submodular_Q} can be developed as
\begin{align*}
  \norm{\bar{\mathbb{T}}(\mathcal{A})}_F^2&-\norm{\bar{\matx{T}}_i(\mathcal{A}_i\cup\{x\})\circ\bar{\mathbb{T}}_{-i}(\mathcal{A})}_F^2\\
  &\geq\norm{\bar{\matx{T}}_i(\mathcal{A}_i\cup\{y\})\circ\bar{\mathbb{T}}_{-i}(\mathcal{A})}_F^2\\
  &\>-\norm{\bar{\matx{T}}_i(\mathcal{A}_i\cup\{x,y\})\circ\bar{\mathbb{T}}_{-i}(\mathcal{A})}_F^2.
\end{align*}
Rewriting this expression using Property~\ref{prop:frob}, we can express the left hand side as
\begin{equation*}
  \braket{\bar{\matx{H}}_i(\mathcal{A}_i),\bar{\mathbb{H}}_{-i}(\mathcal{A})}-\braket{\bar{\matx{H}}_i(\mathcal{A}_i\cup\{x\}),\bar{\mathbb{H}}_{-i}(\mathcal{A})},
\end{equation*}
and the right hand side as
\begin{equation*}
  \braket{\bar{\matx{H}}_i(\mathcal{A}_i\cup\{y\}),\bar{\mathbb{H}}_{-i}(\mathcal{A})}-\braket{\bar{\matx{H}}_i(\mathcal{A}_i\cup\{x,y\}),\bar{\mathbb{H}}_{-i}(\mathcal{A})}.
\end{equation*}
Leveraging the linearity of the inner product we arrive at
\begin{align*}
  &\langle\bar{\matx{H}}_i(\mathcal{A}_i)-\bar{\matx{H}}_i(\mathcal{A}_i\cup\{x\}),\bar{\mathbb{H}}_{-i}(\mathcal{A})\rangle\\
  &\,\,\,\geq\braket{\bar{\matx{H}}_i(\mathcal{A}_i\cup\{y\})-\bar{\matx{H}}_i(\mathcal{A}_i\cup\{x,y\}),\bar{\mathbb{H}}_{-i}(\mathcal{A})}.\label{eq:angle_bound_1}\numberthis
\end{align*}
Developing the matrices using Lemma~\ref{lemma:union_H}, we can operate on both sides of this expression giving, for the left hand side
  \begin{equation*}
    \braket{-\matx{H}_i(\{x\})+2\bar{\matx{T}}_i(\mathcal{A}_i)\circ \matx{T}_i(\{x\}),\bar{\mathbb{H}}_{-i}(\mathcal{A})},
  \end{equation*}
  and for the right hand side
  \begin{equation*}
  \braket{-\matx{H}_i(\{x\})+2\bar{\matx{T}}_i(\mathcal{A}_i\cup\{y\})\circ \matx{T}_i(\{x\}),\bar{\mathbb{H}}_{-i}(\mathcal{A})}.
  \end{equation*}
  Substituting in \eqref{eq:angle_bound_1}, we get
  \begin{align*}
    &\braket{\bar{\matx{T}}_i(\mathcal{A}_i)\circ \matx{T}_i(\{x\}),\bar{\mathbb{H}}_{-i}(\mathcal{A})}\\
    &\geq\braket{\bar{\matx{T}}_i(\mathcal{A}_i\cup\{y\})\circ \matx{T}_i(\{x\}),\bar{\mathbb{H}}_{-i}(\mathcal{A})},
  \end{align*}
and using Lemma~\ref{lemma:union_T_bar} we finally arrive at
  \begin{equation}
    \braket{\matx{T}_i(\{y\})\circ \matx{T}_i(\{x\}),\bar{\mathbb{H}}_{-i}(\mathcal{A})}\geq0,
  \end{equation}
  which is always satisfied because the inner product of positive semidefinite matrices is always non-negative.

Next, suppose  $x\in\mathcal{N}_i$ and $y\in\mathcal{N}_j$ with $i\neq j$, then \eqref{eq:submodular_Q} can be rewritten as
  \begin{align*}
    &\braket{\bar{\matx{H}}_i(\mathcal{A}_i)-\bar{\matx{H}}_i(\mathcal{A}_i\cup\{x\}),\bar{\mathbb{H}}_{-i}(\mathcal{A})}\\
    &\geq\braket{\bar{\matx{H}}_i(\mathcal{A}_i)-\bar{\matx{H}}_i(\mathcal{A}_i\cup\{x\}),\bar{\matx{H}}_j(\mathcal{A}_j\cup\{y\})\circ\bar{\mathbb{H}}_{-(i,j)}(\mathcal{A})}.
  \end{align*}
Using Lemma~\ref{lemma:union_H}, we can further develop this expression into
  \begin{align*}
    &\braket{-\matx{H}_i(\{x\})+2\bar{\matx{T}}_i(\mathcal{A}_i)\circ \matx{T}_i(\{x\}),\bar{\mathbb{H}}_{-i}(\mathcal{A})}\\
    &\geq\left\langle-\matx{H}_i(\{x\})+2\bar{\matx{T}}_i(\mathcal{A}_i)\circ \matx{T}_i(\{x\}),\right.\\
    &\qquad\left.\bar{\matx{H}}_j(\mathcal{A}_j\cup\{y\})\circ\bar{\mathbb{H}}_{-(i,j)}(\mathcal{A})\right\rangle.
  \end{align*}
Leveraging the linearity of the inner product this can be simplified as
  \begin{align*}
    &\left\langle-\matx{H}_i(\{x\})+2\bar{\matx{T}}_i(\mathcal{A}_i)\circ \matx{T}_i(\{x\}),\right.\\
    &\qquad\left.\bar{\mathbb{H}}_{-i}(\mathcal{A})-\bar{\matx{H}}_j(\mathcal{A}_j\cup\{y\})\circ\bar{\mathbb{H}}_{-(i,j)}(\mathcal{A})\right\rangle\geq0.\numberthis\label{eq:last_ineq}
  \end{align*}
Here, we can factorize the left entry of the inner product as
  \begin{align*}
    -\matx{H}_i(\{x\})&+2\bar{\matx{T}}_i(\mathcal{A}_i)\circ \matx{T}_i(\{x\})\\
    &=\matx{T}_i(\{x\})\circ\left[2\bar{\matx{T}}_i(\mathcal{A}_i)-\matx{T}_i(\{x\})\right]\\
    &=\matx{T}_i(\{x\})\circ\left[\bar{\matx{T}}_i(\mathcal{A}_i)+\bar{\matx{T}}_i(\mathcal{A}_i\cup\{x\})\right],\label{eq:last_posdef}\numberthis
  \end{align*}
  which is positive semidefinite due to Property 1, and the fact that the set of positive semidefinite matrices is closed under matrix addition.

  Similarly, the right entry of the inner product in \eqref{eq:last_ineq} can be factorized as
  \begin{align*}
    \bar{\mathbb{H}}_{-i}(\mathcal{A})&-\left(\bar{\matx{H}}_j(\mathcal{A}_j)+\matx{H}_j(\{y\})-2\bar{\matx{T}}_j(\mathcal{A}_j)\circ \matx{T}_j(\{y\})\right)\\
    &\quad\circ\bar{\mathbb{H}}_{-(i,j)}(\mathcal{A})\\
    &=\left(-\matx{H}_j(\{y\})+2\bar{\matx{T}}_j(\mathcal{A}_j)\circ \matx{T}_j(\{y\})\right)\circ\bar{\mathbb{H}}_{-(i,j)}(\mathcal{A}).
  \end{align*}
  The expression inside the parenthesis is analagous to that in \eqref{eq:last_posdef}. Hence, the resulting matrix is positive semidefinite, and thus \eqref{eq:last_ineq} is always satisfied, proving submodularity of $Q$ for all cases.